\documentclass[preprint,10pt]{elsarticle}
\usepackage[displaymath]{preview} 
\usepackage{graphicx}
\usepackage{cancel}
\usepackage{latexsym}
\usepackage{subfigure}
\usepackage{epsfig}
\usepackage{amsmath} 
\usepackage{multirow}
\usepackage{rotating} 
\usepackage{enumerate} 
\usepackage{amssymb}
\usepackage[mathscr]{eucal} 
\usepackage{amsthm} 
\usepackage{units}
\usepackage{color} 
\usepackage[a4paper,text={6.0in,9.0in},centering,
includefoot,foot=0.6in]{geometry} 
\usepackage{pdfsync} 
 
\allowdisplaybreaks[1]

\usepackage{ucs} 
\usepackage[utf8x]{inputenc}
\usepackage[english]{babel} 
\usepackage{fontenc} 
\usepackage{amsfonts}
\usepackage{bm} 
\usepackage{amsthm} 
\usepackage{todonotes}

\newtheorem{theorem}{Theorem}[section]

\newtheorem{proposition}[theorem]{Proposition}
\newtheorem{corollary}[theorem]{Corollary}
\newtheorem{example}[theorem]{Example}

\graphicspath{{./img/}} 
\journal{Journal of Computational Physics}
\begin{document}

\begin{frontmatter}

  \title{Improved convergence of scattering calculations in the
    oscillator representation.}

\author[ua,bitp]{Y.\ Bidasyuk} \ead{Yuriy.Bidasyuk@ua.ac.be}
\author[ua]{W.\ Vanroose} \ead{Wim.Vanroose@ua.ac.be}
\address[ua]{Departement Wiskunde-Informatica, Universiteit Antwerpen,
  Antwerpen, Belgium} \address[bitp]{Bogolyubov Institute for
  Theoretical Physics, Kyiv, Ukraine}
\begin{abstract}
  The Schr\"odinger equation for two and tree-body problems is solved
  for scattering states in a hybrid representation where solutions are
  expanded in the eigenstates of the harmonic oscillator in the
  interaction region and on a finite difference grid in the near-- and
  far--field.  The two representations are coupled through a
  high--order asymptotic formula that takes into account the function
  values and the third derivative in the classical turning points.
  For various examples the convergence is analyzed for various physics
  problems that use an expansion in a large number of oscillator
  states.  The results show significant improvement over the JM-ECS
  method [Bidasyuk et al, Phys. Rev. C \textbf{82}, 064603 (2010)].
\end{abstract}


\end{frontmatter}

\section{Introduction} 
The accurate prediction of the breakup of a many-particle system into
multiple fragments is one of the most challenging problems in quantum
mechanics.  Not only the relative motion of the particles needs to be
modeled, but also the internal structure of the target and the
products need to be described accurately.  This leads in many cases,
to a high--dimensional Schr\"odinger equation posed on a huge domain.
For example, in breakup or collisions of nuclear clusters the cross
section depends on a delicate interplay of the forces that hold the
clusters together and the forces between the clusters.

The quantum state that describes the internal structure of a bound
many particle system is often represented as linear combination of
eigenstates of the quantum mechanical oscillator. These states form an
$L^2$-basis which reduces the problem to finding the correct expansion
coefficients.  Examples are the correlated many electron state of a
quantum dot
\cite{PhysRevB.47.2244,sako2003confined,PhysRevB.80.045321} and the
many nucleon state in nuclear physics
\cite{filippov1980use,Vasilevsky2001a}.  The various applications of
the oscillator representation are discussed in
\cite{dineykhan1995oscillator}. Such a representation is very
efficient for tightly bound ground state or lowest excited states as
any smooth potential well is close to parabolic shape near it's
bottom. As a result, if the oscillator parameter is optimized to match
this local parabolic potential, the lowest states of a system can be
efficiently represented with only a few oscillator functions \cite{navratil}.

However, the representation is inefficient to describe scattering or
break-up processes. Scattering states are not square integrable and
many oscillator states are required to represent the interaction and
asymptotic region.  Furthermore, many potential matrix elements need
to be calculated and this results in a dense linear system that has a
complexity of $N^3$ to arrive at the solution, where $N$ is the
number of oscillator states used in the representation, omitting the
cost of calculating the potential matrix elements.

The $J$-matrix method offers a way of calculating cross sections and
other scattering observables in the oscillator representation. It was
proposed by Heller and Yamani in the seventies \cite{heller1974new,heller} and
mainly applied to atomic problems.  The method exploits the tridiagonal
structure of the kinetic energy operator.  The $J$-matrix has
been under constant development since its inception and a review of
the recent developments can be found in \cite{jmatrix}.  For problems
with Coulomb interactions the Coulomb-Sturmian basis is preferred over
the oscillator representation. Recently the Coulomb-Sturmians have
been used to describe multiphoton single and double-ionization
\cite{foumouo2006theory} and electron impact ionization
\cite{PhysRevA.82.022708}.

At the same time, and mostly independent, the Algebraic Resonating Group Method
was developed for nuclear scattering problems
\cite{filippov1980use,filippov1981,Arickx1994,okhrimenko1984allowance,
PhysRevA.55.265}.
This method exploits the same principles as the $J$-matrix method for
description of nuclear cluster systems, where the oscillator representation is
efficiently used to describe the internal structure of clusters. If the same
representation is used for intercluster degrees of freedom, then
the nucleon symmetrization rules become straightforward.

One shortcoming of the methods based on an $L^2$ basis is that the
asymptotic solutions need to be known explicitly before a system for
the wave function in the inner region and the scattering observables
can be written down.  This is a serious limitation since it is hard to
find the asymptotic wave function for breakup reactions with multiple
fragments.  The Complex Scaling method, which scales the full domain
into the complex plane, can be easily implemented in the oscillator
representation by taking a complex valued oscillator strength. This
has been used to calculate energy spectrum or extract the resonances
\cite{csoto}, but the calculation of cross sections and other
scattering observables may be quite difficult.

In recent decades significant progress has been made in the numerical
solution of scattering processes described by the Helmholtz equation.  In
contrast to many particle systems, where the potential $V$ is often
non-local, the wave number $k(x)$ in the Helmholtz equation depends
only on the local material parameters such as the speed of sound in
acoustics or the electric permittivity and magnetic permeability for
electromagnetic scattering. For these problems
grid based representations such as finite difference, finite
element \cite{NME:NME1620380303,Babuska1995325}, or Discontinuous Galerkin
\cite{Farhat20016455,Farhat20031389} are preferred since
they lead to sparse matrices that can be easily solved by
preconditioned Krylov subspace methods \cite{vandervorst, erlangga2006comparison}.
 
Another technique that has found widespread application is the use of absorbing
boundary conditions. These boundary conditions allow a scattering
calculation without prior knowledge about the asymptotic wave form.
Exterior Complex Scaling (ECS) and Complex Scaling (CS) are widely
used in atomic and molecular physics
\cite{McCurdy2004,moiseyev1998quantum,rescigno1999collisional,vanroose2005complete}.  Perfectly matched layers
(PML) are used for electromagnetic and acoustic scattering
\cite{berenger1994perfectly}, which can also be interpreted as a complex
stretching transformation \cite{chew}.  There are many other excellent
absorbing boundary conditions \cite{givoli,antoine, nissen2010perfectly}.

In \cite{PhysRevC.82.064603} the JM-ECS method was introduced that
combines the $J$-Matrix method with a grid based ECS.  The method
describes the scattering solution in the interior region with an
oscillator representation and in the exterior region with finite differences.
The two representations are matched through a low order asymptotic
formula with an error that scales as $\mathcal{O}(N^{-1/2})$, where
$N$ is the size of the oscillator basis describing the inner region.
Once the grid and oscillator representation are matched, it is easy
to introduce an absorbing boundary layer since the grid representation
can be easily extended with an ECS absorbing layer or any other
absorbing boundary condition.

The resulting method was illustrated for one-- and
two--dimensional model problems representative for real scattering
problems with local interactions.  Furthermore, the representation was
used for nuclear $p$-shell scattering. However, the accuracy of the
calculations was unsatisfactory due the low--order matching condition.
While the grid representation on the exterior has an accuracy of
$\mathcal{O}(N^{-1})$, the matching is only accurate to order
$\mathcal{O}(N^{-1/2})$.

The main contribution of the current paper is to increase the accuracy
of the asymptotic formula that allows a better matching of the grid
and oscillator representations.  The better asymptotic formula takes
function values and its third derivative into account. It can bring
the matching error down to the level of the accuracy of the grid
representation.
  
A higher--order asymptotic approximation of the oscillator
representation was already discussed by S. Igashov in
\cite{igashov_jmatrix}.  However, the formula was not used to
increase the accuracy of scattering calculations.

This paper is outlined as followed.  In Section
\ref{sec:reviewscattering} we shortly describe the process of scattering
calculations in the oscillator representation.  In Section
\ref{sec:higherorder} we derive a higher--order asymptotic formula that
takes into account the behavior of the function in the
classical turning points in coordinate and Fourier space.  In Section
4 and 5 we use this asymptotic formula to solve scattering
problems.

\section{Review of scattering calculations in the oscillator representation}
\label{sec:reviewscattering} In this section we discuss the most
important properties of the oscillator representation and how they can
be used to perform scattering calculations with the $J$-matrix method.
We also recall the working of the hybrid $J$-matrix and ECS method
proposed in \cite{PhysRevC.82.064603}.
\subsection{The radial scattering equation} 
The aim is to solve the Schr\"odinger equation in atomic units ($m=1$,
$\hbar=1$) that describes a scattering process of two particles for an energy $E\in \mathbb{R}$.  This
equation written in relative coordinates is
\begin{equation}\label{eq:scattering3d}
\left[-\frac{1}{2}\Delta + V(\mathbf{r}) - E \right] \Psi(\mathbf{r}) = F(\mathbf{r}),  \quad \forall \mathbf{r} \in \mathbb{R}^3
\end{equation}
where $F(\mathbf{r})$ is the function describing the initial state or
the source term and $V(\mathbf{r})$ is the potential.  The coordinate
$\mathbf{r}$ can be written in spherical coordinates $(\rho, \theta,
\phi)$. In case of spherically symmetric potential
($V(\mathbf{r})=V(\rho)$), Eq.~\eqref{eq:scattering3d} can be
reduced to a one-dimensional radial equation using partial wave
decomposition
$$
\Psi(\mathbf{r})=\Psi(\rho, \theta, \phi) = \sum_{l,m} \frac{\psi_l(\rho)}{\rho} Y_{l,m}(\theta, \phi), \quad 
F(\mathbf{r}) = F(\rho, \theta, \phi) = \sum_{l,m} \frac{f_l(\rho)}{\rho} Y_{l,m}(\theta, \phi),
$$
where $Y_{l,m}(\theta, \phi)$ is a spherical harmonic, $l=0,1,2,
\ldots$ is the orbital angular momentum of the relative motion, $m$ is
the projection of this angular momentum.  The resulting reduced radial
equation becomes
\begin{equation}\label{eq:scattering}
\left[-\frac{1}{2}  \frac{d^2}{d \rho^2} + \frac{l(l+1)}{2 \rho^2} + V(\rho) - E \right] \psi_l(\rho) =
f_l(\rho).
\end{equation}

For the problem we are interested in there is a range $a > 0 $ such
that for all $\rho>a$ both $V(\rho)$ and $\chi(\rho)$ are zero.

We solve the Eq.~\eqref{eq:scattering} for $E>0$ and extract from the
solution $\psi_l$ scattering observables such as the cross sections or
phase shift. In order to solve the Eq.~\eqref{eq:scattering} we represent
the solution as
\begin{equation} \label{eq:expansion}
\psi_l(\rho) = \sum_{n=0}^{\infty} c_{n,l} \varphi_{n,l}(\rho),
\end{equation}
where $\varphi_{n,l}(\rho)$ are orthogonal $L^2$ functions, in particular we
will use reduced oscillator functions, whose properties will be explained in
the next section.  After projection of~\eqref{eq:expansion} on $\varphi_{n,l}$, Eq.~\eqref{eq:scattering} results
in an infinite linear system
\begin{equation}\label{eq:discretescattering}
  \sum_{n=0}^\infty \left(T^{(l)}_{kn} + V^{(l)}_{kn} - E \right)c_{n,l} =b_{k,l},
\end{equation} where  $T_{kn}^{(l)}$ denotes the elements of the kinetic energy matrix and $V_{kn}$ the elements of the potential energy matrix.  They are the integrals 
\[T^{(l)}_{kn} = \int_0^\infty \varphi_{k,l}(\rho) \left(-\frac{1}{2} 
\frac{d^2}{d \rho^2} + \frac{l(l+1)}{2\rho^2} \right) \varphi_{n,l}(\rho) \; d\rho,\]
\[V^{(l)}_{kn} = \int_0^\infty \varphi_{k,l}(\rho) V(\rho)\varphi_{n,l}(\rho)
\; d\rho
\]
and $b_{k,l} = \int_0^\infty \varphi_{k,l}(\rho) \chi_l(\rho) \; d\rho$. As the considered problem is effectively one-dimensional we will use $x$ instead of $\rho$ to denote radial relative coordinate in all following sections devoted to two-body problem.

\subsection{The oscillator representation} Before we explain the
strategy to solve Eq.~\eqref{eq:discretescattering}, we repeat
the main properties of the oscillator representation and the function
$\varphi_{n,l}$ that will be used in the expansion.

The reduced radial equation analogous to~\eqref{eq:scattering} for the quantum harmonic oscillator with an
oscillator strength $\omega$ is
\begin{equation}\label{eq:oscil_rad}
 \left[-\frac{1}{2} \frac{d^2}{d x^2} +\frac{1}{2}\frac{l(l+1)}{x^2} + \frac{1}{2}\omega^2 x^2
\right] \varphi_{n,l}(x) = E_{n,l} \varphi_{n,l}(x)
\end{equation} with $x \in [0, \infty[$ a radial coordinate. The
boundary conditions are $\varphi_{n,l}(0)=0$ and $\lim_{x\rightarrow \infty}
\varphi_{n,l}(x) = 0$. The eigenvalues are
\begin{equation} \label{eq:osc_spectrum}
E_{n,l} = \left(2 n + l + \frac{3}{2}\right)\omega,
\end{equation}
and eigenstates
\begin{equation}\label{eq:oscillatorfunction}
 \varphi_{n,l}(x) = (-1)^n N_{n,l} b^{-1/2} \left(\frac x b\right)^{l+1}
\exp\left(-\frac{x^2}{2b^2}\right) L_n^{l+1/2}\left(\frac{x^2}{b^2}\right),
\end{equation}
where $L_n^{l+1/2}$ are Laguerre polynomials. The normalization is
$N_{n,l} = \sqrt{2n!/\Gamma(n+l+3/2)}$, where $n \in \{0,1,\ldots,\}$
and oscillator length is defined as $b=\sqrt{1/\omega}$.

The classical turning point associated with each state is $R_{n,l} = b
\sqrt{4n + 2l + 3}$, defined as the point where the potential energy
equals the total energy of the system. 

The functions (\ref{eq:oscillatorfunction}) form a complete
orthonormal basis and any wave function $\psi_l(x)$ that behaves as $x^l$ in $x=0$
can be represented using the infinite sum:
\begin{equation} 
  \psi_l(x) = \sum_{n=0}^{\infty} c_{n,l}
  \varphi_{n,l}(x), \quad \text{where} \quad c_{n,l} = \int_0^\infty
  \varphi_{n,l}(x) \psi_l(x) dx.
 \label{eq:osc_rep}
\end{equation} 
In the following section we will use that the radial oscillator
equation, \eqref{eq:oscil_rad}, can be rewritten in terms of $b$, the
oscillator length, and $R_{n,l}$, the classical turning point as
\begin{equation}\label{eq:oscillator}
  \left[-\frac12 \frac{d^2}{dx^2}
    + \frac{l(l+1)}{x^2}  + \frac{x^2}{2 b^4}- \frac{R_{n,l}^2}{2 b^4}\right]
  \varphi_{n,l}(x)=0.
\end{equation}

An important property that forms the basis of the results in this
paper is that the $n$-th oscillator state has $n$ oscillations between
the origin and the classical turning point
$R_{n,l}=\sqrt{4n+2l+3}$. Beyond this turning point the function is
exponentially decaying  without additional oscillations. This
means that as $n$ increases, the frequency of the oscillation between
the origin and $R_{n,l}$ grows proportional to $\sqrt{n}$. This property
will be used to derive the asymptotic formula in Section \ref{sec:higherorder}.
\begin{figure}
  \begin{center}
	\includegraphics[width=0.5\linewidth]{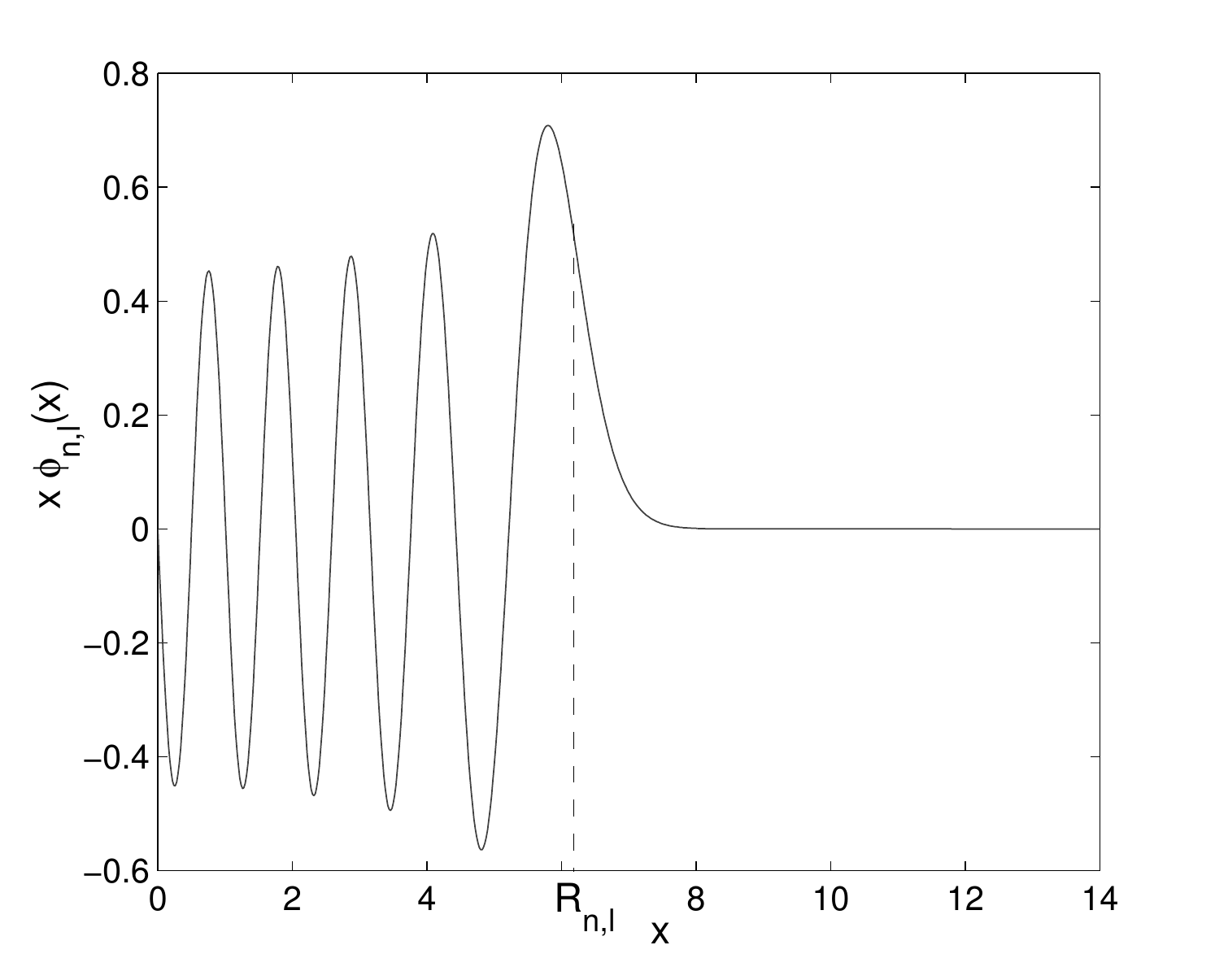}\includegraphics[
width=0.5\linewidth]{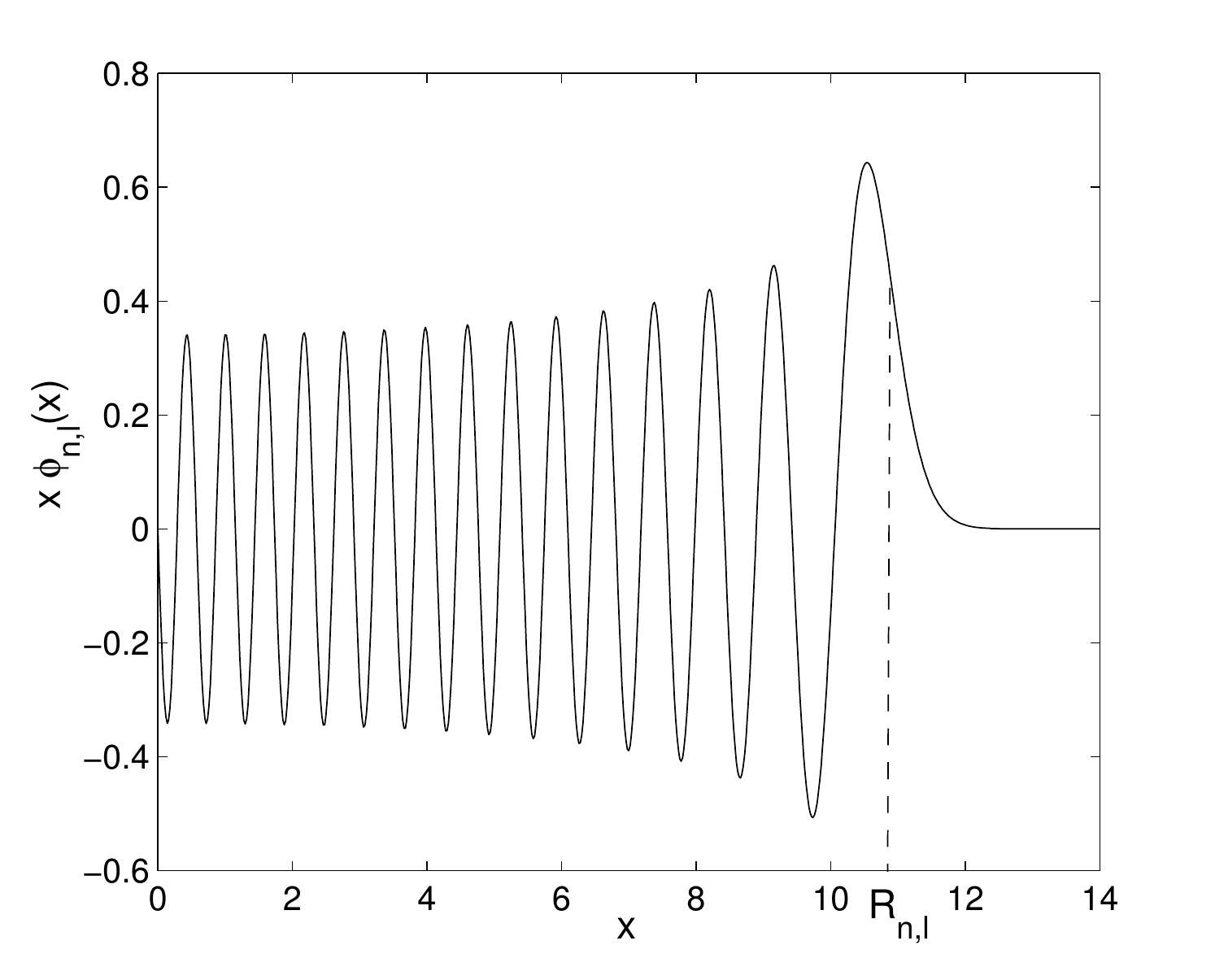}
      \end{center}
      \caption{The reduced oscillator state $\varphi_{n,l}(x)$ with a classical turning point $R_{n,l}$ for
        $l=0$ and two values of $n$: $n=10$ (left) and $n=30$
        (right). The function oscillates $n$ times on the interval
        $[0,R_{n,l}]$. Since $R_{n,l}$ only grows with
        $\mathcal{O}(\sqrt{n})$ the state becomes rapidly oscillating as
        $n$ grows.}
	\label{fig:oscil}
\end{figure}

The Bessel transform $\tilde{F_l}(k)$ of a function $F_l(x)$ is defined
as
\begin{equation} \tilde F_l(k) := \sqrt{\frac2{\pi}}
\int_0^\infty F_l(x) \hat j_l(kx) dx.
\end{equation} 
where $\hat j_l(kx)$ is a Riccati--Bessel function, the regular
solution of the radial Schr\"odinger equation without potential,
\eqref{eq:scattering}. It is connected with the ordinary Bessel function
in a following way $\hat j_l(x) = \sqrt{\pi x/2}\,J_{l+1/2}(x)$.
The Bessel transform of an oscillator state is again an
oscillator state
\begin{equation} \tilde \varphi_{n,l}(k) = (-1)^n b\,
\varphi_{n,l}(kb^2),
\end{equation} 
with a classical turning point $K_{n,l} = R_{n,}/b^2 = \sqrt{4n+2l+3}/b$. This
relation that can easily be derived using {\em e.g.}  formula
(7.421.4) from \cite{gradshtein}. The oscillator states therefore form
an orthonormal basis of $L^2([0,\infty[)$ that diagonalizes the Bessel
transform. Similar properties hold for the Cartesian oscillator state
based on Hermite polynomials, see, for example,
\cite{PhysRevB.80.045321}.

An other important property that will be used in the remainder of the
text is the following:
\begin{proposition}\label{lem:interchange}
  Let $\psi_l$ a function that behaves as $x^l$ in $x=0$. The projection on the oscillator
  state $\varphi_{n,l}$ can be calculated with either $\psi_l$ or its
  Bessel transform $\tilde{\psi}_l$ as:
\begin{equation} 
  c_{n,l} = \int_0^\infty \varphi_{n,l}(x) \psi_l(x) dx =
  (-1)^n b \int_0^\infty \varphi_{n,l}(x) \tilde{\psi}_l(x) dx.
\end{equation}
\end{proposition}
\begin{proof} Since Parseval's theorem holds we can calculate $c_n$
either with $\varphi_{n,l}$ or its Bessel transform
\begin{equation} 
  c_{n,l} = \int_0^\infty \varphi_{n,l}(x) \psi_l(x) dx
  = \int_0^\infty \tilde{\varphi}_{n,l}(k) \tilde{\psi}_l(k) dk.
\end{equation}
But since $\tilde{\varphi}_{n,l}(k)= (-1)^n b \varphi_{n,l}(k b^2)$,
the oscillator state in the latter integral is the same as in the
first integral, after substitution of variables and up to a
constant.\end{proof}
This result will be used in the next section to
derive the asymptotic formula for $c_{n,l}$ with large $n$.  This
symmetry between $\psi_l(x)$ and its Bessel transform $\tilde{\psi_l}(k)$
should
also hold in the asymptotic formula.

The kinetic energy operator $T_{i,j}^l$ is a tri-diagonal matrix. For the
oscillator basis the non-zero elements are
\begin{equation}
  T^l_{i,j} = \int_0^\infty \varphi_{i,l}(x)
  \left(-\frac{1}{2}\frac{d^2}{dx^2} + \frac{l(l+1)}{2x^2}\right)
  \varphi_{j,l}(x) dx 
  = \left\{\!\!\begin{array}{c@{\quad}l}
      \left(2i+l+\frac32 \right)\frac{\omega}{2} & \text{for } j=i,\\
      -\sqrt{i\,\left(i+l+\frac12\right)}\frac{\omega}{2} & \text{for }
      j=i-1,\\ -\sqrt{(i+1)\left(i+l+\frac32 \right)}\frac{\omega}{2}&
      \text{for } j=i+1.
\end{array}\right.
\label{eq:kinetic_en}
\end{equation}

For the remainder of the text we will consider only the case of zero angular
momentum and drop $l$ from
the notation. All results, however, are valid for arbitrary $l$.

\subsection{The asymptotic formula} In \cite{Vanroose2001} an
asymptotic formula was derived for the projection of a smooth radial function
$\psi$ on the state $\varphi_n$. The derivation uses stationary phase
arguments to exploit the increase of oscillations as $n\rightarrow
\infty$. We repeat here the derivation of the results.

There are two main contributions to the value of the integral
\begin{equation} 
  c_n = \int_0^\infty \varphi_n(x)\psi(x) dx \approx I_0 + I_{R_n}.
\end{equation}
There is a first contribution, denoted $I_0$, of the left integration
boundary near zero. There is a second contribution, denoted
$I_{R_n}$, from the point of stationary phase, which is the classical
turning point where the oscillations stop (see Figure \ref{fig:oscil}).
Here we provide only general steps of calculating this integral in the
asymptotic region. More details can be found in \cite{Vanroose2001}.

The contribution from the classical turning point can be calculated by
approximating the
oscillator state $\varphi_n$ near the classical turning point by an
Airy function as
\begin{equation}
  \varphi_n(x) \approx   \frac{2}{b}\left(\frac{b^4}{2R_n}\right)^{1/6}
\mathrm{Ai}\left[
    \left(\frac{2R_n}{b}\right)^{1/3}(x-R_n)\right] \quad \text{if} \quad
  |x-R_n| \ll 1.
\end{equation} 
The integral representation of this Airy function and the stationary
phase approximation leads to the contribution of the turning point to
the integral
\begin{equation} 
  I_{R_n} \approx b\sqrt{\frac{2}{R_n}} \psi(R_n).
\end{equation} 
The contribution of the left integration boundary,
$I_0$, can be derived by approximating the oscillator state near the origin by
a Riccati--Bessel function 
\begin{equation}
 \varphi_n(x) \approx (-1)^n \frac{\sqrt{2}}{b}\sqrt{\frac{2}{\pi K_n}}
\hat j_0(K_n x),
\end{equation}
where $K_n$ is the classical turning point of the oscillator state in
the momentum space. Then the contribution from the origin becomes
\begin{equation}
  I_0 \approx (-1)^n\frac{1}{b}\sqrt{\frac{2}{K_n}}
  \tilde{\psi}(K_n).
\end{equation}
And the resulting asymptotic approximation of the oscillator coefficients
becomes
\begin{equation}\label{eq:prl_firstorder}
c_n \approx (-1)^n \frac{1}{b} \sqrt{\frac{2}{K_n}} \tilde{\psi}(K_n) + b
\sqrt{\frac{2}{R_n}} \psi(R_n) \quad \text{if} \quad
  n \gg 1.
\end{equation} 
This relation has a contribution from the turning points in the
coordinate space, $R_n$ and the Fourier space, $K_n$.  Note that it
satisfies the symmetry observed in \ref{lem:interchange} above.


Note that Eq.~\eqref{eq:prl_firstorder} as well as the similar equation
in \cite{Vanroose2001} does not provide the order of the
approximation.  This is one of the shortcomings that will be addressed
in the current paper.
%

\subsection{Scattering calculations in the oscillator representation}
We now discuss how a finite linear system can be obtained that solves
for the wave function in the interaction region and the phase shift,
describing the solution in the asymptotic region.  The presentation
here is based on the asymptotic formula and differs from how the
method was derived historically.  For more detail about the formulation of the original
$J$-matrix method we refer to \cite{heller} and \cite{Arickx1994}.

Let $\psi(x)$ be a smooth two-body radial scattering state with a bounded
energy, depending on one spatial coordinate (can be always reduced
using center-of-mass relative coordinates and spherical coordinates).
Since it is a scattering state, the function does not go to zero as
$x\rightarrow \infty$. However, its Bessel transform $\tilde{\psi}(k)$
goes to zero as $k \rightarrow \infty$.  This means that, as $n$
grows, the only contribution to the expansion coefficient $c_n$ comes
from the classical turning point in coordinate space. So, for a
scattering state $\psi$ with total energy $E=k^2/2$ the solution is 
asymptotically $\hat{j}_l(kr) + \tan (\delta_l)
\hat{n}_l(kr)$ where $\delta_l$ is the phase shift~\cite{Tay83} holds that
\begin{equation}
  c_n \approx b \sqrt{\frac{2}{R_n}} \psi(R_n) =  
b\sqrt{\frac{2}{R_n}}\left(\hat{j}_l(kR_n) + \tan (\delta_l) 
\hat{n}_l(kR_n)\right),
  \label{eq:as_coeff}
\end{equation} where $\hat{j}_l$ and $\hat{n}_l$ are the regular and
irregular solutions of the free-particle equation.  The $c_n$ becomes
the representation of $\psi$ on the grid of classical turning points
$R_n$.

The aim of a one-dimensional radial scattering calculation is to find a
numerical approximation to $\tan (\delta_l)$ for a given potential $V$.  This
can be achieved by
writing the solution in the oscillator representation as
\begin{equation}\label{eq:asymptotic} 
c_{n,l} = \begin{cases}
c^0_{n,l} + j_{n,l} + \tan (\delta_l) \, n_{n,l} & n<N \\ j_{n,l} + \tan(\delta_l)\,
n_{n,l} & n \geq N,
  \end{cases}
\end{equation} where $j_{n,l}$ is the regular solution of the
homogeneous three-term recurrence relation
\begin{equation} 
T_{n,n-1}j_{n-1,l} + (T_{n,n}-E)j_{n,l} + T_{n,n+1}j_{n+1,l} = 0,
\end{equation} where as $n \gg 1$ holds that $j_{n,l}
=b\sqrt{2/R_n}\,\hat{j}_l(kR_n)$.  The $n_{n,l}$ is the irregular
solution of the recurrence relation that goes as $n_{n,l} =
b\sqrt{2/R_n}\,\hat{n}_l(kR_n)$ when $n$ becomes large.

With the form \eqref{eq:asymptotic} we reduce the infinite linear
system to a finite dimensional problem with a set of $N+1$ unknowns
$\{c^0_{i,l}, \tan \delta_l \}$. Once solved we simultaneously obtain the wave
function of the system, $\{ c^0_{i,l}\}$, and the scattering
information, $\tan (\delta_l)$.

  A detailed description of the construction of this linear system,
known as the $J$-matrix, can be found in \cite{Arickx1994}.

\subsection{The hybrid $J$-matrix and ECS method} In
\cite{PhysRevC.82.064603} the asymptotic formula
\eqref{eq:prl_firstorder} was used to introduce a hybrid oscillator
and grid representation for the Schr\"odinger equation that is useful
for scattering calculations where an asymptotic form such
as~\eqref{eq:as_coeff} is not explicitly known. Such a situation
appears in breakup reactions of three or more particles.

Let $\psi(x)$ be again a smooth one-dimensional radial scattering state
with a bounded energy such that
\begin{equation} 
c_n \approx b \sqrt{\frac{2}{R_n}} \psi(R_n).
  \label{eq:matching}
\end{equation}
holds. The $c_n$ becomes the representation of $\psi$ on the grid of
classical turning points $R_n$. The grid distance between these points
becomes smaller when $n$ increases since
\begin{equation}
 h_n = R_{n}-R_{n-1} \approx b^2\frac2{R_n}.
	\label{eq:vgridsize}
\end{equation} 
However, the asymptotic $\psi$ does not necessarily need to be
represented on the grid of turning points $R_n$. Another option is, for example,
a regular grid of equally spaced points. 

The hybrid JM-ECS method represents the one-dimensional radial wave function
as a vector ${\sf \Psi}$ in $\mathbb{C}^{N+K}$, where
\begin{equation} \label{eq:hybridvector} {\sf \Psi} = (c_0,c_1, \ldots,
  c_{N-1},\psi(R_{N}),\psi(R_{N}\! +\! h),\ldots,\psi(R_{N}\! +\!
  (K\!-\!1)h) ).
\end{equation}
The first $N$ elements represent the wave function in the oscillator
representation. While the remaining $K$ elements represent $\psi$ on
an equidistant grid that starts at $R_N$, the $N$-th classical turning
point, and runs up to $R_N + (K-1)h$ with a grid distance $h$ equal to
the difference of the last two turning points of the oscillator
representation $h = R_{N}-R_{N-1}$.  It is assumed that the matching
point that connects the oscillator to finite--difference
representation corresponds to a large index $N$ such that the
asymptotic formula, \eqref{eq:matching}, applies.

Again, the kinetic energy operator in this hybrid representation is
tridiagonal since in both finite--difference and oscillator
representations it is tridiagonal. One should only be careful in 
matching both representations. To obtain the kinetic energy in the
last point of the oscillator representation, the tridiagonal kinetic
energy formula~\eqref{eq:kinetic_en} is used. It involves a recurrence
relation connecting the three terms $c_{N-2}$, $c_{N-1}$ and
$c_{N}$. The latter, the coefficient $c_N$, is unknown. Only
$\psi(R_N)$ is available on the grid.  Using the asymptotic relation~\eqref{eq:matching}, however, it is possible to calculate the required
matrix element as follows:
\[\begin{split} (Tc)_{N-1} = T_{N-1,N-2} c_{N-2} + T_{N-1,N-1} c_{N-1}
+ T_{N-1,N} b \sqrt{2/ R_{N}} \psi_l(R_{N}).
\end{split}\]

To calculate the kinetic energy in the first point of the finite
difference grid, the second derivative of the wave function has to be
known.  To approximate the latter with a finite difference formula,
one needs the wave function in the grid points $R_{N-1}$, $R_{N}$ and
$R_N + h$. Again it is possible to apply~\eqref{eq:matching} to obtain
$\psi(R_{N-1})$ in terms of $c_{N-1}$:
\[ \psi^{\prime\prime}(R_N) \approx \frac{c_{N-1}/( b \sqrt{2/ R_{N-1}}) - 2
\psi(R_{N}) + \psi(R_{N}+h)}{h^2}.
\]

The coupling between both representations around the matching point is
sketched in Figure \ref{fig:coupling}, together with the terms involved
to determine the correct matching.
\begin{figure}[ht]
\begin{center} \hspace{10pt} Oscillator \hspace{30pt} Finite
Differences
$$
\overbrace{\vphantom{\rule[5pt]{5pt}{5pt}}c_{n-2} \;\; \quad c_{N-1}
\qquad \psi(R_N)}^{T c_{N-1}} \quad \psi(R_{N}+h)
\makebox[0pt][r]{$\underbrace{ \phantom{c_{N-1} \qquad \psi(R_N) \quad
\psi(R_{N+1}) \vphantom{\rule[-20pt]{20pt}{20pt}} }}_{\psi''(R_N)}$}
$$
\ifx\JPicScale\undefined\def\JPicScale{0.6}\fi \unitlength \JPicScale
mm
\begin{picture}(120,30)(-1,-32) \linethickness{0.3mm}
\put(0,20){\line(1,0){120}} \linethickness{0.3mm}
\put(90,20){\circle*{3}}

\linethickness{0.3mm} \put(20,20){\circle{3}}

\linethickness{0.3mm} \put(40,20){\circle{3}}

\linethickness{0.3mm} \put(67,20){\circle*{3}}

\linethickness{0.1mm} \put(53,-6){\line(0,1){50}}
\end{picture}
\end{center} \vspace{-60pt}
\caption{Illustration on the calculation of the kinetic energy matrix
  elements that are calculated in the last point $R_{N-1}$ of the
  oscillator representation and in the first point $R_N$ of the finite
  difference representation.  To calculate $T$ applied to a solution
  vector we need to translate the oscillator representation to the
  grid.  Vice versa for the application of the finite difference
  stencil approximation to the second
  derivative. \label{fig:coupling}}
\end{figure}
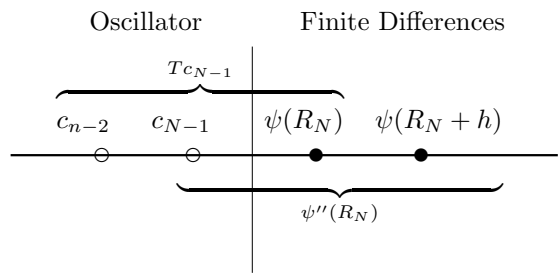

However, in the next section we will show that the asymptotic matching
condition that is used in \cite{PhysRevC.82.064603} to couple both
representations, is only accurate
to $\mathcal{O}(R_N^{-1})$, while the outer grid is accurate to order
$\mathcal{O}(h^2)=\mathcal{O}(R_N^{-2})$.  Therefore the largest error is made
in
the matching condition that couples both representations.

Note that introducing an absorbing layer is easy once the oscillator
representation is coupled to a grid representation.  For example ECS is
implemented by extending the grid with a complex scaled part
\cite{McCurdy2004}.

\section{Higher--order asymptotic formula}
\label{sec:higherorder}
In this section we
derive a higher--order asymptotic formula. It not only takes into
account the function values in the turning point  $R_n$ but also the third
derivative in this point.  A similar asymptotic formula was derived by S.
Igashov in
\cite{igashov_jmatrix}, however, it did not include the contributions
from the Fourier space. 

\begin{proposition}
\label{lemma1} 
Let $\psi(x)$ be a regular scattering state that behaves as $x^l$ in
$x=0$, that is infinitely differentiable and $\varphi_n(x)$ is a reduced
oscillator state, solution of~\eqref{eq:oscillator}.  Then the
projection of the scattering state on the oscillator state can be
approximated as
\begin{equation}\label{eq:first_corr} 
  c_{n} = \int_0^\infty \,  \varphi_n(x)\, \psi(x) dx = b \sqrt{\frac2{R_{n}}} \left( \psi(R_{n})    + \frac{b^4}{6 R_{n}} \psi^{\prime\prime\prime}(R_{n}) \right) +  \mathcal{O}(R_{n}^{-5/2}).
\end{equation} 
This expresses the expansion coefficient in terms of the wave function
and its derivatives in the classical turning point $R_{n}$.
\end{proposition}

\begin{proof} The oscillator Eq.~\eqref{eq:oscillator} can be
rewritten as
\begin{equation} \left[-\frac{d^2}{dx^2} + \frac{2 R_n (x-R_n)}{b^4}
+ \frac{(x-R_n)^2}{b^4}\right] \varphi_n(x)=0.
\label{eq:osc_rewritten}
\end{equation} For $R_n \gg |x-R_n|$ (which means either $n\gg1$ or
$|x-R_n| \approx 0$), we can neglect the quadratic term in this
equation and get the Airy equation with the solution
\begin{equation} \varphi_n(x) \approx \varphi_n^{(0)}(x) = \frac 2 b
\left( \frac{b^4}{2 R_{n}} \right)^{1/6} \mathrm{Ai} \left[
\left(\frac{2 R_{n}}{b^4}\right)^{1/3}(x-R_{n})\right],
\label{eq:airy_1st}
\end{equation} where the normalization is chosen such that it
coincides with the oscillator state near the classical turning point
$R_n$.

This can be further improved by writing $\varphi_n(x) = \left(1 +
u_n(x)\right)\varphi^{(0)}_n$ and inserting this in the oscillator
equation.  The equation for $u_n(x)$ becomes
\begin{equation} \label{eq:33}
  -\frac12 u_n^{\prime\prime} \varphi_n^{(0)} -   u_n^\prime {\varphi_n^{(0)}}^\prime -\frac12   u_n\,{\varphi_n^{(0)}}^{\prime \prime} + \left(R_n(x-R_n) +     \frac{1}{2}(x-R_n)^2\right) u_n \varphi_n^{(0)} =  -\frac{1}{2}(x-R_n)^2 \varphi_n^{(0)}.
 \end{equation} 
 Using that $${\varphi_n^{(0)}}^\prime= \left(2 R_n/b^4\right)^{1/3} 
\mathrm{Ai}^\prime\left[\left(2   R_{n}/b^4\right)^{1/3}(x-R_{n})\right]$$
 and $${\varphi_n^{(0)}}^{\prime\prime}= \left(2 R_n/b^4\right)^{2/3}
 \mathrm{Ai}^{\prime\prime}\left[\left(2   
R_{n}/b^4\right)^{1/3}(x-R_{n})\right]$$
the equation~\eqref{eq:33} becomes
\begin{align}
  &-\frac12 u_n^{\prime\prime} \varphi_n^{(0)} - u^\prime  
\left(\frac{2R_n}{b^4}\right)^{1/3} \mathrm{Ai}^\prime\left[\left(\frac{2
R_{n}}{b^4}\right)^{1/3}(x-R_{n})\right]\nonumber\\
 &-\frac12 u_n  \, \left(\frac{2 R_n}{b^4}\right)^{2/3}
\mathrm{Ai}^{\prime\prime}\left[\left(\frac{2
R_{n}}{b^4}\right)^{1/3}(x-R_{n})\right] + \left(R_n(x-R_n) +
\frac{1}{2}(x-R_n)^2\right) u_n \, \varphi_n^{(0)} = -\frac{1}{2}(x-R_n)^2
\varphi_n^{(0)}.
 \end{align} 
 In the limit $n \rightarrow \infty$, the term with $R_n(x-R_n)$
 dominates compared to $(1/2)(x-R_n)^2$ and the other terms since both
 the first and second derivative of $\mathrm{Ai}(x)$ around zero are
 bounded.  We find that $u_n(x) \approx -(x-R_n)/(2 R_n)$ is a
 solution of the remaining equation.

 At this time we will not consider the contributions to the integral for
 the boundary at zero.  This contribution will be equal to the
 contribution from the Fourier space as in Eq.~\eqref{eq:prl_firstorder} and for scattering states and large $n$ this 
 contribution is negligible. So we can lower the integration boundary and write
\begin{equation} 
  c_{n} \approx
  \frac{2\sqrt{a_n}}{b}\int_{-\infty}^{\infty} \psi(x+R_{n}) \mathrm{Ai}
  (x/a_n) \left(1-\frac{x}{2R_n}\right) dx,
 \label{eq:airy_approx}
\end{equation}
where $a_n:=(b^4/2R_{n})^{1/3}$ and we have substituted the integration
variable $x\rightarrow x+R_n$.

The value of this integral will be determined by the behavior 
around $R_n$, the point of stationary phase. Since the function
$\psi$ is infinitely differentiable we can Taylor expand
\begin{equation}
  c_{n} \approx  \frac{2\sqrt{a_n}}{b}\int_{-\infty}^{\infty}
\sum_{m=0}^{\infty}
  \frac{x^m}{m!}\psi^{(m)}(R_{n}) \left(1-\frac{x}{2R_n}\right)  \mathrm{Ai}
(x/a_n) dx
\end{equation}
All integrals in these series can be calculated explicitly using
specific properties of Airy function (see page 52 of \cite{airy}).
\begin{equation} 
  \int_{-\infty}^{\infty} \mathrm{Ai}(x) x^{3k} dx = \frac{(3k)!}{3^k k!} \quad \text{and} \quad 
  \int_{-\infty}^{\infty} \mathrm{Ai}(x) x^{3k+1} dx =  \int_{-\infty}^{\infty} \mathrm{Ai}(x) x^{3k+2} dx = 0. 
\end{equation}
We finally get
\begin{equation}\label{eq:ho_matching} 
\begin{split}
  c_{n} &\approx \frac{2 a_n^{3/2}}{b} \sum_{k=0}^\infty \frac{a_n^{3k}}{3^k
k!}\left( \psi^{(3k)}(R_{n}) -  \frac{a_n^3}{2R_n}\psi^{(3k+2)}(R_{n})\right) \\
    &= b \sqrt{\frac2{R_{n}}} \sum_{k=0}^{\infty}
    \frac{1}{k!}\left(\frac{b^4}{6 R_{n}}\right)^k
    \left(\psi^{(3k)}(R_{n}) -
      \frac{b^4}{4R_n^2}\psi^{(3k+2)}(R_{n})\right).
\end{split}
\end{equation} 
In lowest order (in terms of $1/R_n$) this gives us exactly the
initial relation (\ref{eq:matching}). Including the next order
correction we get a relation (\ref{eq:first_corr}). For testing
purposes we also include terms of the next order $1/R_n^2$
\begin{equation}
  \begin{split} 
    c_{n} &= b \sqrt{\frac2{R_{n}}} \left(\psi(R_{n}) + \frac{b^4}{6 R_{n}}
\psi^{\prime\prime\prime}(R_{n})\right. \left.
+\frac1{R_{n}^2}\left(\frac{b^8}{72}\psi^{(6)}(R_{n}) -
        \frac{b^4}{4}\psi^{\prime\prime}(R_{n})\right) +
      \mathcal{O}(\frac1{R_n^3})\right)
  \end{split}  \label{eq:second_corr}
\end{equation}
\end{proof}
Note that a similar result is readily obtained for Cartesian harmonic
oscillator states that are based on the Hermite polynomials. There,
however, there is a contribution from both turning points, one from
$R_n=\sqrt{2n + 1}$ and $R_n = -\sqrt{2n+1}$.

\begin{corollary} Let $\psi(x)$ be a function that behaves as $x^l$ in $x=0$ that is infinitely
  differentiable, then the asymptotic expansion coefficient is
\begin{align} 
c_{n} = &\int_0^\infty \,\varphi_n(x)\, \psi(x) dx = b \sqrt{\frac2{R_{n}}}
\left( \psi(R_{n}) + \frac{b^4}{6 R_{n}} \psi^{\prime\prime\prime}(R_{n})
\right)\nonumber \\
 & + \frac{(-1)^n}b \sqrt{\frac2{K_{n}}} \left( \tilde{\psi}(K_{n}) +
\frac{1}{6 b^4 K_{n}} {\tilde{\psi}}^{\prime\prime\prime}(K_{n}) \right) +
\mathcal{O}(n^{-5/4}).
\label{eq:withfourier}
\end{align}
\end{corollary}
\begin{proof} The asymptotic formula should have the same result when
  we interchange $\psi(x)$ with its Bessel transform $\tilde{\psi}(k)$
  as a result of Proposition \ref{lem:interchange}. 
 Indeed, the
  integral 
  \[
  c_n  = (-1)^n b \int_0^\infty \varphi_n(k)\tilde{\psi}(k)dk
\]
will have a contribution from integral boundary near 0 and a
contribution from the turning point $K_n$ in Fourier space.  The
latter is, with the help of the previous results, $(-1)^n
b\sqrt{2/K_n}(\tilde{\psi}(K_n) +
\tilde{\psi}^{\prime\prime\prime}(K_n)/6b^4K_n)$.  While the
contribution from boundary become the contribution from the turning
point $R_n$ in coordinate space.

\end{proof}

\begin{example} We illustrate the convergence of the asymptotic
  formula \eqref{eq:withfourier} with an application to the function $\psi(x) =
\exp(-a
  x)\sin(x)$.  We give results for two choices of $a$. As the scale of
  the representation is defined by the oscillator length $b$, the
  result will change with its choice.  If the product $a b$ is large,
  the Fourier components will dominate in the expansion
  coefficient. When $a b$ is small the function $\psi$ resembles a
  scattering state, on a scale defined by $b$, and the coordinate
  approximation will dominate. However, as Figure~\ref{fig:fourier}
  illustrates the combined formula always gives the correct result. We
  see an overall convergence with $n^{-5/4}$. This is smaller than
  $\mathcal{O}(h_n^2)$ which goes as $\mathcal{O}(1/n)$. In general, it is
possible to construct a function, for which both coordinate and Fourier
terms fail to represent the exact oscillator coefficient, but the combined
formula
remains valid even in this case (see Figure~\ref{fig:fourier2}).

\begin{figure}
  \begin{center}
    \begin{tabular}{cc}
      \includegraphics[width=0.5\linewidth]{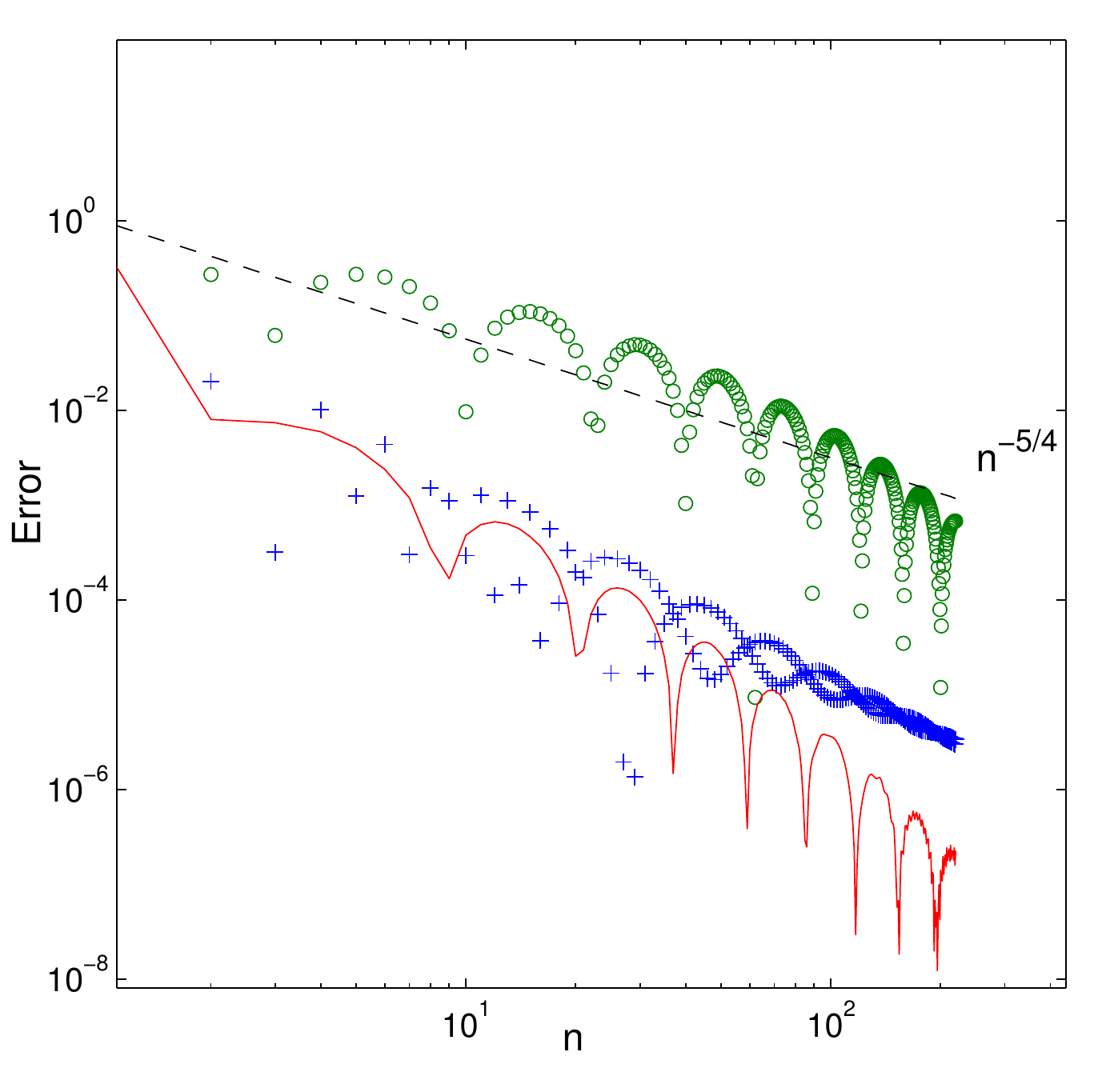}&\includegraphics[width=0.5\linewidth]{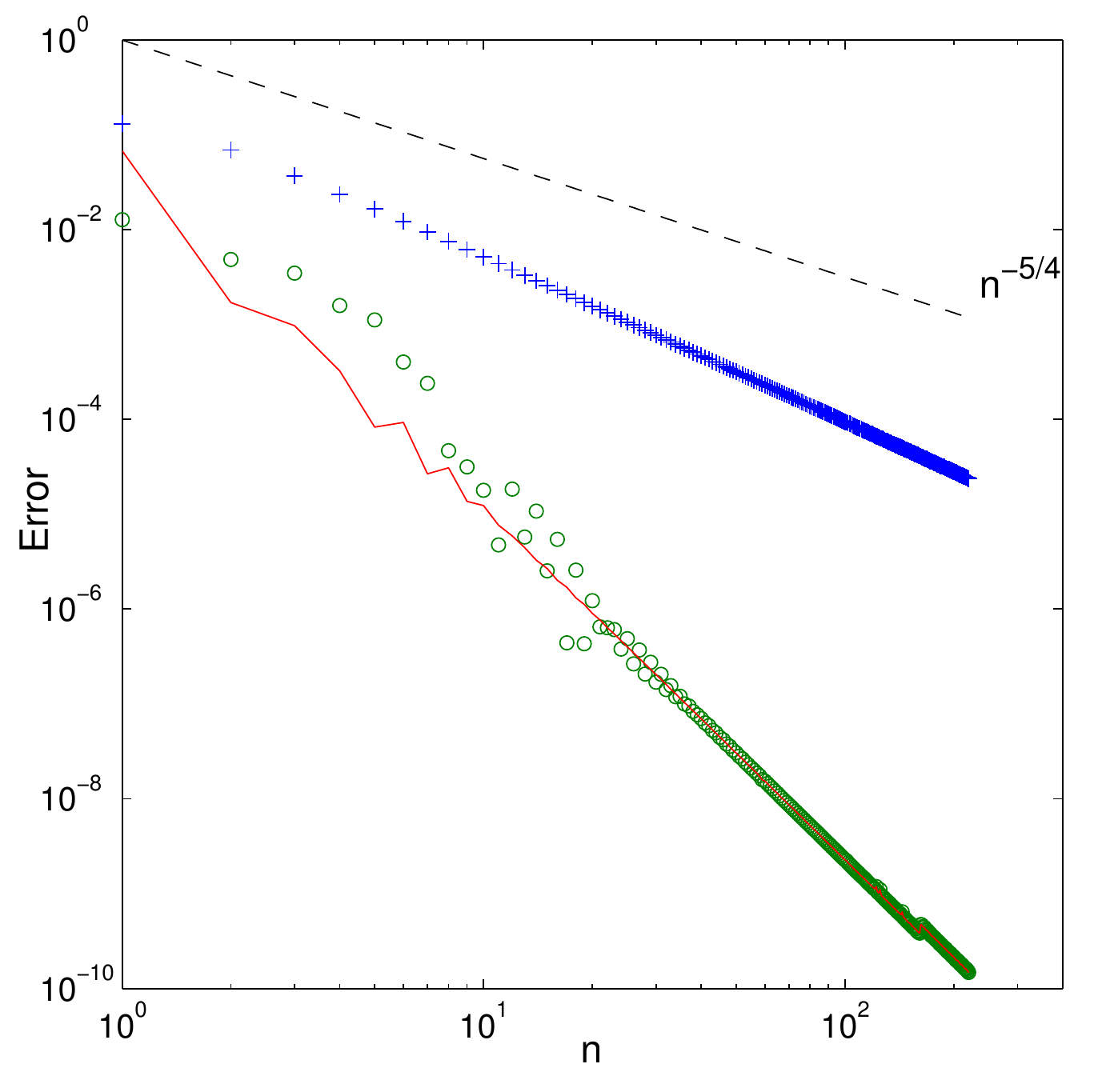}
    \end{tabular}
  \end{center}
  \caption{ Convergence of the general asymptotic formula
(\ref{eq:withfourier}) with only coordinate terms (crosses), only
Fourier terms (circles), both coordinate and Fourier terms (solid
line) for test function $\psi(x) = \exp(-a x)\sin(x)$ with different
values of $a$: left figure --- $a=0.2$, right figure --- $a=1.5$
($b=1$ in both cases)}
  \label{fig:fourier}
\end{figure}

\begin{figure}
  \begin{center}
\includegraphics[width=0.5\linewidth]{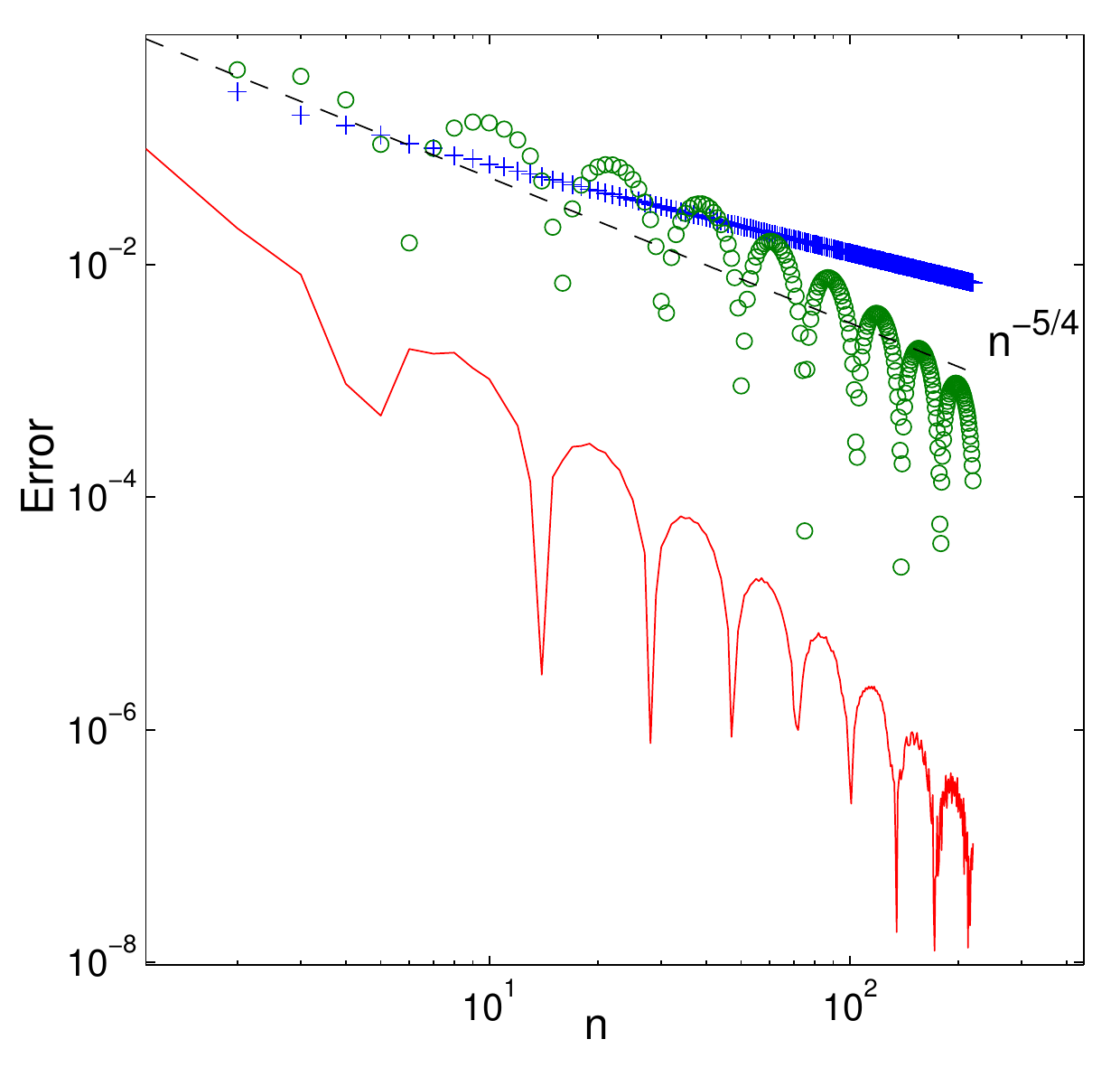}
  \end{center}
  \caption{ Convergence of the general asymptotic formula
    (\ref{eq:withfourier}) for the test function $\psi(x) = \exp(-a
    x)\cos(x)$ with $a=0.2$, for which neither coordinate (crosses)
    nor Fourier terms (circles) give accurate result, while the
    combined formula (solid line) does.}
  \label{fig:fourier2}
\end{figure}

\end{example}
\begin{example} We have verified the asymptotic relations derived in
  the previous section with a numerical experiment using the wave
  function $\psi(x)=\sin(kx)$. We calculate oscillator coefficients
  directly and compare them against the values obtained with
  asymptotic formulas of different order derived previously (see
  Fig.~\ref{fig:coeff_as}).
\begin{figure}
  \begin{center}
	\includegraphics[width=0.5\linewidth]{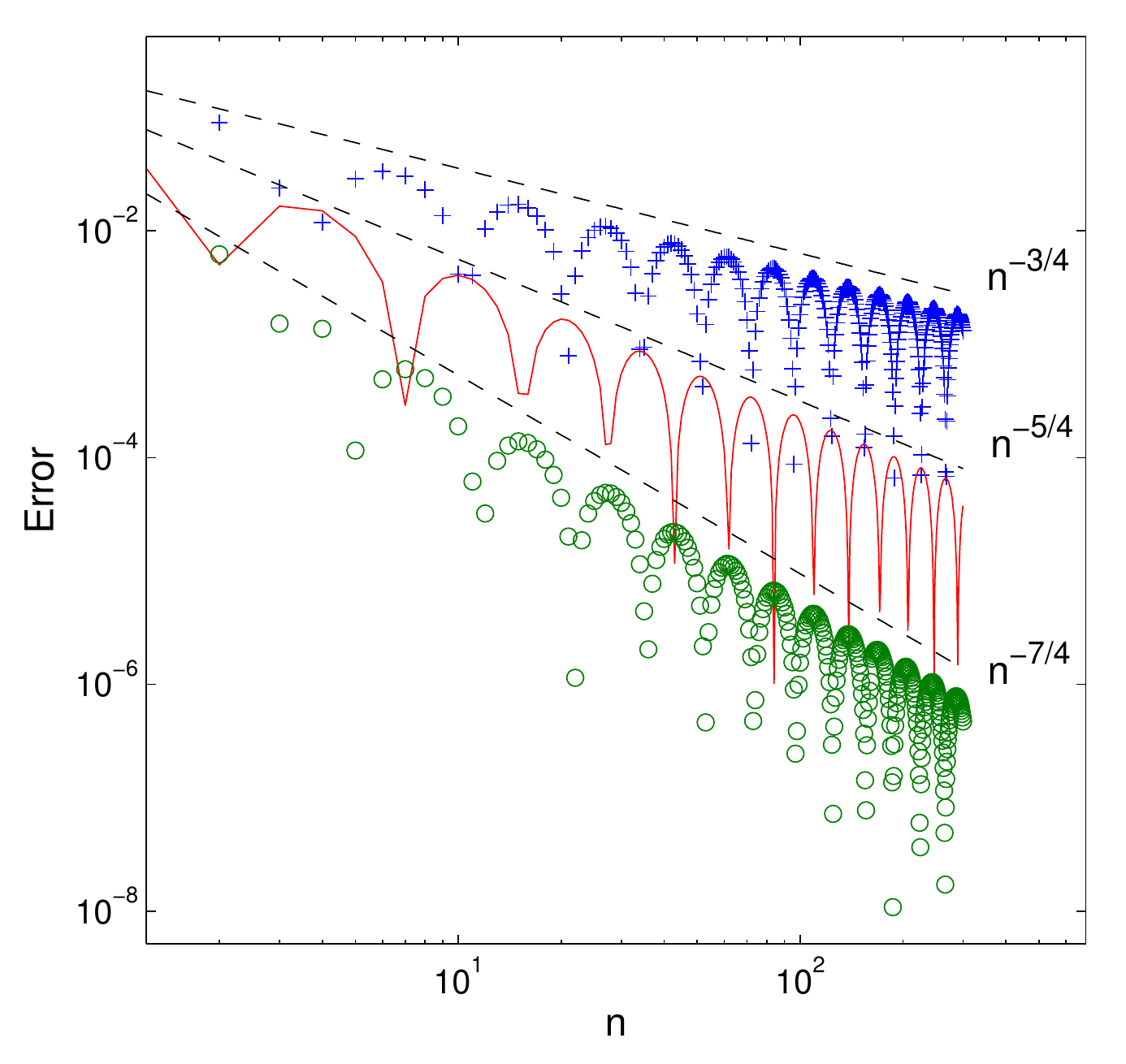}\includegraphics[
width=0.5\linewidth]{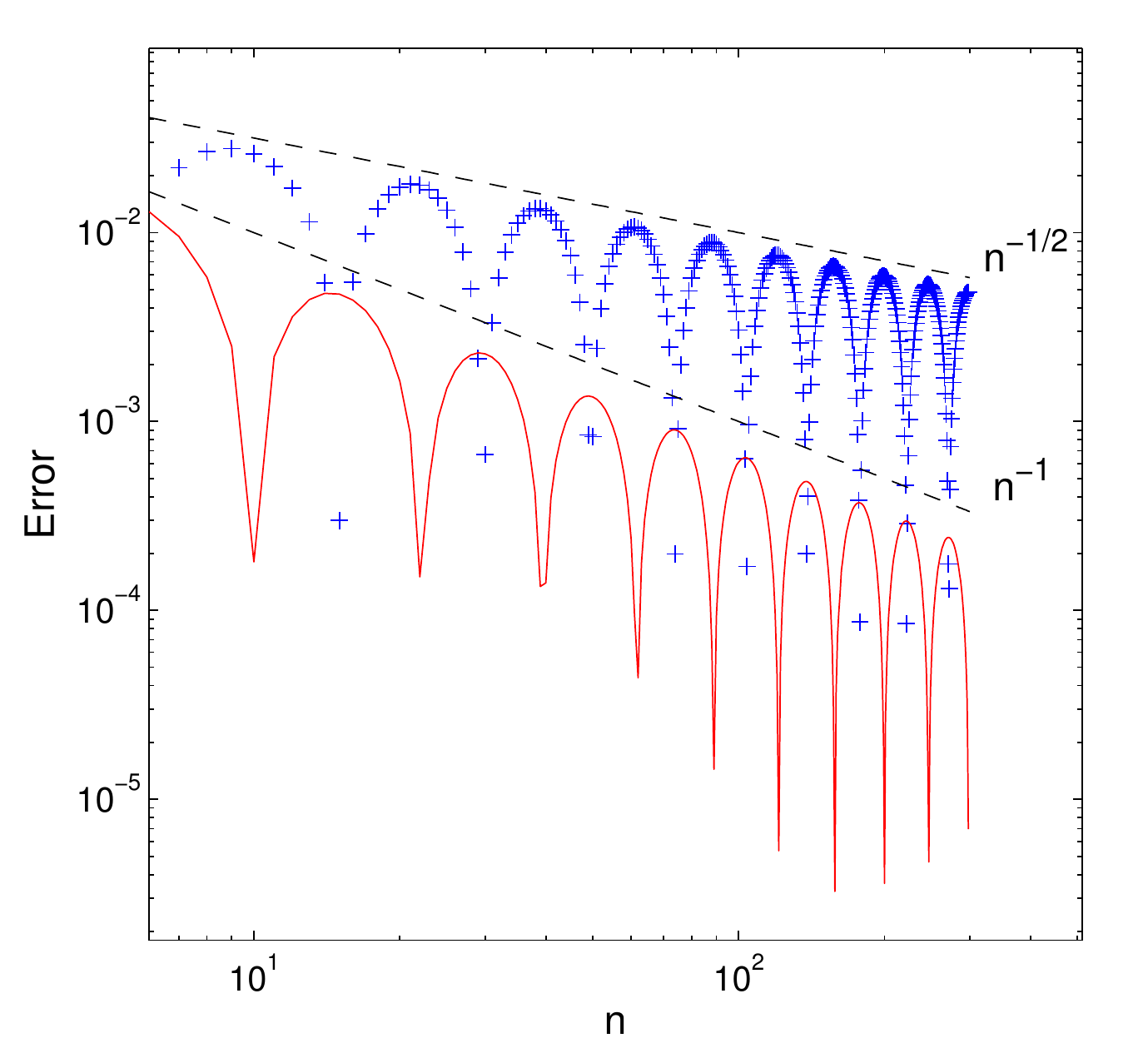}
      \end{center}
      \caption{Convergence of the asymptotic formula
(\ref{eq:second_corr}) (left) and the inverse asymptotic formula
(\ref{eq:first_corr_inv}) (right) with different number of terms included
(crosses --- one term, solid line --- two terms, circles --- three terms). The
test function has the form $\psi(x)=\sin(kx)$ with $k=1.5$ and the oscillator
length $b=0.8$}
\label{fig:coeff_as}
\end{figure}
\end{example}

\subsection{Inverse relation} 
The asymptotic expansion coefficient is an approximation with the help
of the functions values of $\psi$ evaluated at certain grid points.
It is also useful to derive in inverse relation that allows us to
construct the function value in certain grid points given the
expansion coefficients.

To approximate $c_n$ with Eq.~\eqref{eq:first_corr} the function
value and its third derivative need to be calculated at the turning
point $R_n$. When only the function values of $\psi$ are available at
the turning points $R_n$, the third derivative can be approximated by
finite differences.

The coefficient is then calculated as
 \begin{equation} c_{n} = b \sqrt{\frac2{R_{n}}} \left(
\psi(R_{n}) + \frac{b^4}{6 R_{n}} \sum_{k} D^{(3)}_{k}
\psi(R_{n+k}) \right) + \mathcal{O}(R_n^{-5/2}) + R_n^{-3/2} \epsilon_n, 
 \label{eq:first_corr_fd}
\end{equation}
where $D^{(3)}$ is the matrix with coefficients that approximates the
third derivative and $k$ indicates the stencil points and $\epsilon_n$
is the error term of this approximation. In case of 5-point finite
difference approximation $k\in\{-2,-1,0,1,2\}$.  Note that the grid of
turning points $R_n$ is an irregular grid and the coefficients will
depend on the local distances between the neighboring grid points. The
error of the approximation $\epsilon_n$ will also depend on the local
grid distances as $\mathcal{O}(h^{p}) \sim \mathcal{O}(R_n^{-p})$,
where $p$ is the order of the approximation.

The resulting transformation can be represented by a banded sparse matrix $U$
\begin{equation}
  c_{n}  = \sum_{k} U_{nk} \psi(R_k) +  \mathcal{O}(R_n^{-5/2}).
\end{equation}
Note that the matrix elements of $U$ will only depend on the values of
$R_n$. Indeed, The differential operator $D^{(3)}$ only depends on the
grid distances and so do the coefficients of the asymptotic
expression. 

The relation \eqref{eq:first_corr_fd} translates $\psi$ on the grid of
turning points to corresponding $c_n$.  It is also useful to derive
the inverse relation that gets $\psi(R_n)$ from known values of
$c_n$. We can not use a direct inversion of Eq.~\eqref{eq:first_corr} since it involves the third derivative, a dense
operator in the oscillator representation. However, an approximate
inverse relation can be obtained by rearranging terms in
(\ref{eq:first_corr_fd}) as
 \begin{equation} 
   \psi(R_{n}) = \frac1b \sqrt{\frac{R_{n}}2} c_{n} - \frac{b^4}{6 R_{n}}
\sum_{k} D^{(3)}_{k} \psi(R_{n+k}) + \mathcal{O}(R_n^{-2}) + R_n^{-1} \epsilon_n
\end{equation} 
and replacing values of the wave function in the right--hand side by
oscillator coefficients using only the first term of
(\ref{eq:matching}), i.e.  $\psi(R_n) = (1/b)\sqrt{R_n/2}\,c_n +
\mathcal{O}(R_n^{-1})$.  This only introduces an error of the order of
$\mathcal{O}(R_n^{-1})$ but combined with the $1/R_n$ this gives the
approximate inverse relation
 \begin{equation} 
   \psi(R_{n}) = \frac1b \sqrt{\frac{R_{n}}2}    c_{n} - \frac{b^3}{6
R_{n}} \sum_{k} D^{(3)}_{k}   \sqrt{\frac{R_{n+k}}2} c_{n+k} +
\mathcal{O}(R_n^{-2})
 \label{eq:first_corr_inv}
\end{equation} 
that is accurate to $\mathcal{O}(R_n^{-2})$.  An example of the
resulting numerical accuracy of this inverse asymptotic relation is
shown on the right panel of figure~\ref{fig:coeff_as}.  Also this
transformation can be presented as a sparse banded matrix
multiplication
\begin{equation}\label{eq:w_matrix}
  \psi(R_n)  = \sum_{n} W_{nk} c_{k}  + \mathcal{O}(R_n^{-2}).
\end{equation}
Again the matrix elements of $W$ only depend on the values of the turning points $R_n$. 

Combining the two relations leads to an approximate partition of
unity:
\begin{equation}
U W = I + \mathcal{O}(R_n^{-2}).
\end{equation}
It is important to note that since these transformation matrices only
depend on values of $R_n$ they can, in principle, both be defined for
an arbitrary grid.

\subsection{Approximate discretization of Operators}
With the help of these two transformation matrices $U$ and $W$ we can
now build an approximate oscillator representation of an operator $Q$
using its finite difference representation. Let $Q^{(fd)}$ be the
finite difference representation of $Q$ on the grid $R_n$. Then the
application of the $Q$ on $\psi$ can be written as 
\begin{equation*} 
(Q \psi)^{(osc)}_n = \sum_{m}Q^{(osc)}_{nm} c_m = \sum_{m}\check Q^{(osc)}_{nm} c_m + \mathcal{O}(n^{-1}), \qquad
\text{where} \quad \check Q^{(osc)}_{nm} := [U Q^{(fd)} W]_{nm}.
\end{equation*} 
We illustrate the accuracy of this relation on a second
derivative operator  $D^(2)$, as we intend to apply the constructed
representation to Helmholtz-type equations.

\begin{equation} \check D^{(2)(osc)}=U D^{(2)(fd)} W,
\label{eq:d2trans}
\end{equation} 
where $D^{(2)(fd)}$ is a finite difference matrix of the second
derivative. To analyze the accuracy of this approximation
Figure~\ref{fig:app_der} shows the result of the error operator
defined as $(D^{(osc)}-\check D^{(osc)})$ acting on the vector of
oscillator coefficients of the test function $\psi(x)=\sin(kx)$. For
comparison we show the error of the approximate identity operator,
that can be defined as $(I - \check I) = (I - U W)$. We see that both
considered error operators decay as $\mathcal{O}(n^{-1})$ in high-$n$
region. This means that our approximate operators are asymptotically
equal to the exact ones and we can expect the same order of
convergence in a solution of the scattering problem.

\begin{figure}
\begin{center}
	\includegraphics[width=0.5\linewidth]{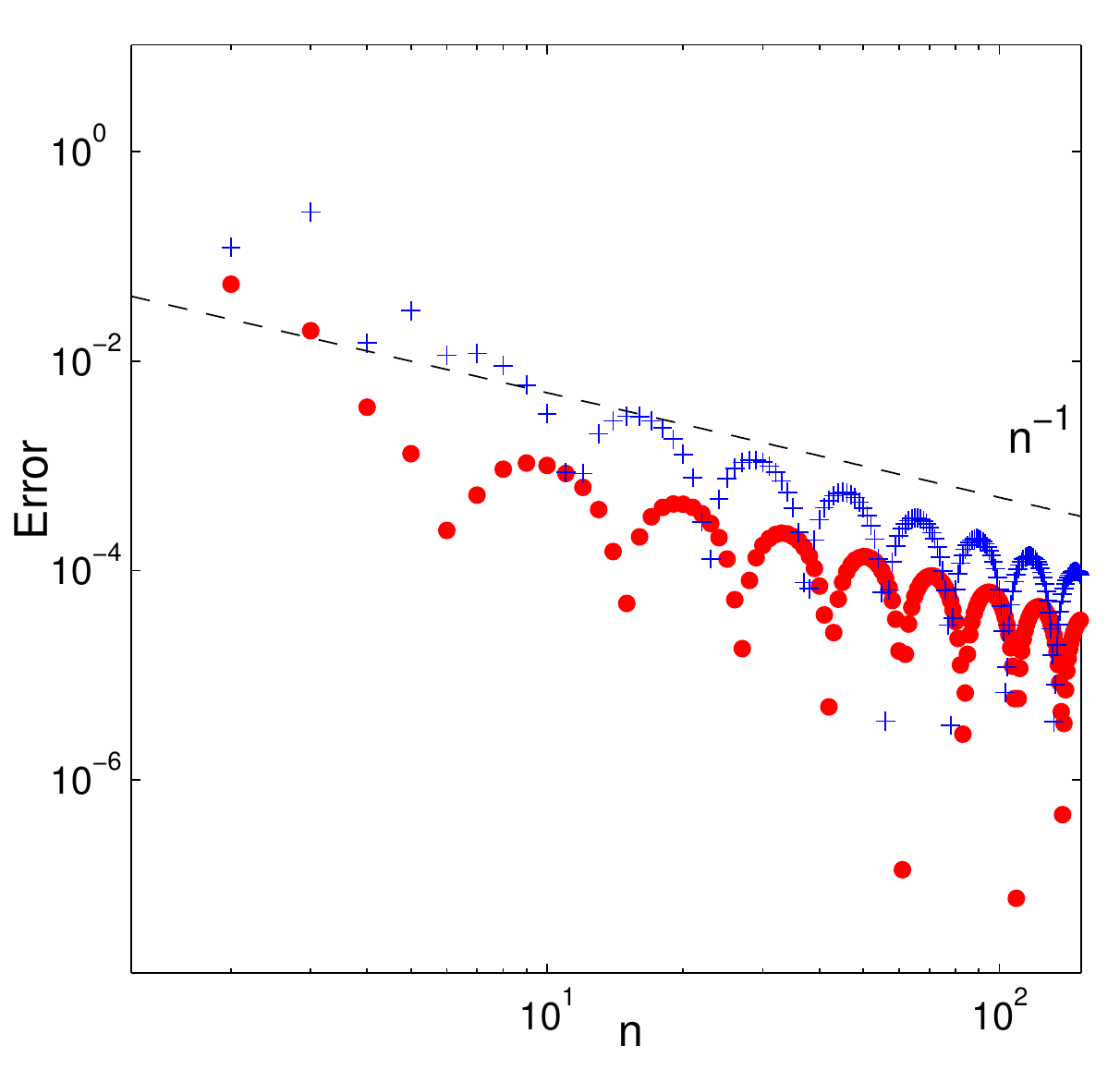}
      \end{center}
      \caption{Convergence of the approximate second derivative
operator $\check D = U D^{(fd)} W$ (crosses) and the approximate identity
operator $\check I = U W$ (circles) with a test function of the form
$\psi(x)=\sin(kx)$ ($k$ and $b$ are the same as in Figure~\ref{fig:coeff_as})}
\label{fig:app_der}
\end{figure}  

It is important to note that all matrices in (\ref{eq:d2trans})
can be built for any arbitrary spatial grid, though it will not be an
approximate oscillator representation. But we can modify this grid
only in the asymptotic region ({\em e.g.} to implement the absorbing boundary
with ECS). Then the coefficients $c_n$ have
no physical meaning as oscillator coefficients but the coordinate wave
function can sill be reconstructed with (\ref{eq:w_matrix}).

\section{High-order hybrid representation for scattering problems}
The aim is now to build a hybrid representation for scattering
calculations where the oscillator basis is used in the internal region
and finite differences in the outer region.  The matching of the two
should use the higher--order asymptotic formula \eqref{eq:first_corr}.
However, the strategy displayed on Fig.~\ref{fig:coupling} cannot be
easily applied since the third derivative of the wave function needs
to be estimated from several neighboring points symmetrically
distributed on both sides of the matching point.  But in the matching
point we only have the required function values on one side of the
matching point.  This arguments holds both for the oscillator and for
the finite difference representation.  We therefore present a matching
strategy that differs from the proposal of \cite{PhysRevC.82.064603}.

Consider a arbitrary grid $\bar{R}_n$ in $[0,\infty]$ with
$n=1,2,\ldots$ that differs from the grid of turning points $R_n$.  On
this grid we can construct the differential operator $D^{(3)}$ and
construct the operators $\bar{U}$ and $\bar{W}$ by replacing every
appearance of $R_n$ in equations \eqref{eq:first_corr_fd} and
\eqref{eq:first_corr_inv} by $\bar{R}_n$. Similarly, for a function
$\psi$ sampled in the grid points $\bar{R}_n$ we can calculate
\begin{equation}
  \bar{c}_n  = \sum_{k} \bar{U}_{nk} \psi(\bar{R}_k).
\end{equation} 
These coefficients $\bar{c}_n$ are not the expansion coefficients of
the function $\psi$ in the oscillator function basis.  Only when
$\bar{R}_n$ equals $R_n$, the oscillator turning points, then the
$\bar{c}_n$ are approximations to the oscillator expansion
coefficients $c_n$.

To build the hybrid representation, we choose the grid $\bar{R}_n$
such that first $N+1$ points correspond to the turning points
$R_n$. The remaining points of $\bar{R}_n$, for $n > N+1$, are
chosen to correspond to points on a equidistant finite difference
grid.  To solve a scattering problem the equation
\eqref{eq:scattering} needs to be discretized with finite differences
on the grid $\bar{R}_n$ and then transformed with the help of
$\bar{U}$ and $\bar{W}$ into
\begin{equation}
\sum_j \left(  \bar{U}H^{fd}\bar{W}-E \bar{U}\bar{W}\right)_{ij} \bar{c}_j = (\bar{U}f)_i
\end{equation}
to arrive at an equation for $\bar{c}_i$.

The first $N$ coefficients of $\bar{c}_n$ correspond now to
approximations to oscillator expansion coefficients $c_n$.  However,
because $n$ is low, they are only a poor approximation. The idea is
now to replace the first $N$ coefficients with the exact coefficients.
At the same time we replace the first $N$ rows of the matrix with the
exact operator in the oscillator representation.  The linear system
then becomes
{\footnotesize
\begin{equation}
 \left(\begin{array}{ccc|cc}
 H^{(osc)}_{00} - E &    \hdots&    H^{(osc)}_{0N}  &H^{(osc)}_{0N+1} &\hdots\\ 
 H^{(osc)}_{10} &   \hdots  & H^{(osc)}_{1N} & H^{(osc)}_{1N+1} & \hdots \\ 
\vdots        &             &      \vdots       & \vdots   &\\
 H^{(osc)}_{N0} &   \hdots  & H^{(osc)}_{NN}-E & H^{(osc)}_{NN+1} & \hdots \\
\hline
[\bar{U}H^{(fd)}\bar{W}]_{N+1,0} & \hdots & [\bar{U}H^{(fd)}\bar{W}]_{N+1,N} & [\bar{U}(H^{(fd)}-E)\bar{W}]_{N+1,N+1} & \hdots\\
\vdots                           &&\vdots&\vdots&
\end{array} \right)  \left(\begin{array}{c} \vphantom{H^{(N)}_N} c_{0}\\
\vphantom{H^{(N)}_N} c_{1}\\
\vdots\\
\vphantom{H^{(N)}_N} c_{N}\\
 \hline 
\vphantom{H^{(N)}_N} \bar{c}_{N+1}\\
\vdots 
\end{array} \right) = \left(\begin{array}{c} f^{(osc)}_{0}\\ 
f^{(osc)}_{1}\\ 
\vdots\\
f^{(osc)}_{N}\\ 
\hline
\sum_j\bar{U}_{N+1j} f(x_j)\\
\vdots
\end{array} \right),
\label{eq:combined}
\end{equation} 
}
where $H^{(osc)}$ is the representation of the Hamiltonian in the
oscillator representation and $H^{(fd)}$ in the finite difference
representation.

We emphasize the difference with \cite{PhysRevC.82.064603}. Here we do
not match two regions by using the asymptotic formula in one
point. Now the representation in the asymptotic region is a fairly
good approximation of the oscillator representation. Therefore the
structure of the discretized wave function is simpler than
\eqref{eq:hybridvector}. Now we have
$$
\bar{\sf \Psi} = \left(c_0, \ldots c_N, \bar{c}_{N+1}, \ldots \bar{c}_{K}\right),
$$ 
where in the initial version of hybrid method, the latter elements of
the vector are the function values of $\psi$ in the grid points. Now 
the internal region is covered by an exact oscillator representation,
and the asymptotic region is covered by approximate representation
which is based on finite differences and includes the ECS
transformation.

\subsection{Numerical Illustration} 
We illustrate the method for a one-dimensional radial example. After the
solution of (\ref{eq:combined}), we first reconstruct the coordinate
wave function using (\ref{eq:w_matrix}). Outside the range of the
potential $V$ the solution can be written as a linear combination
$\psi_{sc} = A \hat{h}_l^+(kx) + B \hat{h}_l^-(kx)$, where
$\hat{h}_l^\pm$ are the in- and outgoing Riccati-Hankel functions. The
coefficient $A$ is then extracted with
\begin{equation} A =
W\left(\psi_\text{sc}(x),\hat{h}_l^-(kx)\right)/W\left(\hat{h}_l^+(kx),\hat{h}_l^-(kx)\right),
\end{equation} 
where $x$ is outside the range of the potential but
still on the real part of the ECS domain. The Wronskian is calculated
as $W(u,v) = u^\prime v-v^\prime u$, where the derivatives can be
implemented with finite differences.  From the Wronskians for $A$ and
$B$ we can extract the phase shift of the solution.

It is important to note that the accuracy of the method does not
directly depend on the size of the asymptotic region as long as this
region is located outside the range of the potential ($V=0$ in the
considered part of $\check H^{(osc)}$). In this region we use a grid
of 150 equidistant points and then the ECS layer that spans 10
dimensionless length units and applies a complex rotation of
$45^{\circ}$ to the coordinate axis.

We first consider a model problem with a attractive potential in
Gaussian form, $V(x) = -\exp(-x^2)$, and we limit our consideration to
the case of zero total angular momentum $l=0$. Nevertheless, all the
conclusions are applicable to higher angular momenta as well.  The
right panel of Figure \ref{fig:pshift} shows the error in the
scattering phase shift of the considered problem calculated with the
original JM-ECS and the new approach. To obtain the reference phase
shift we use the highly accurate variable phase approach (VPA). We see
that for all energies the new method gives much more accurate and less
oscillatory results. The convergence as a function of the number of
oscillator states in the inner region of both methods is shown on
Figure \ref{fig:pshift_conv}. We see that we arrived at the desired
the convergence rate to $N^{-1}$ in both low-energy and high-energy
regions.

 \begin{figure}
  \includegraphics[width=0.41\linewidth]{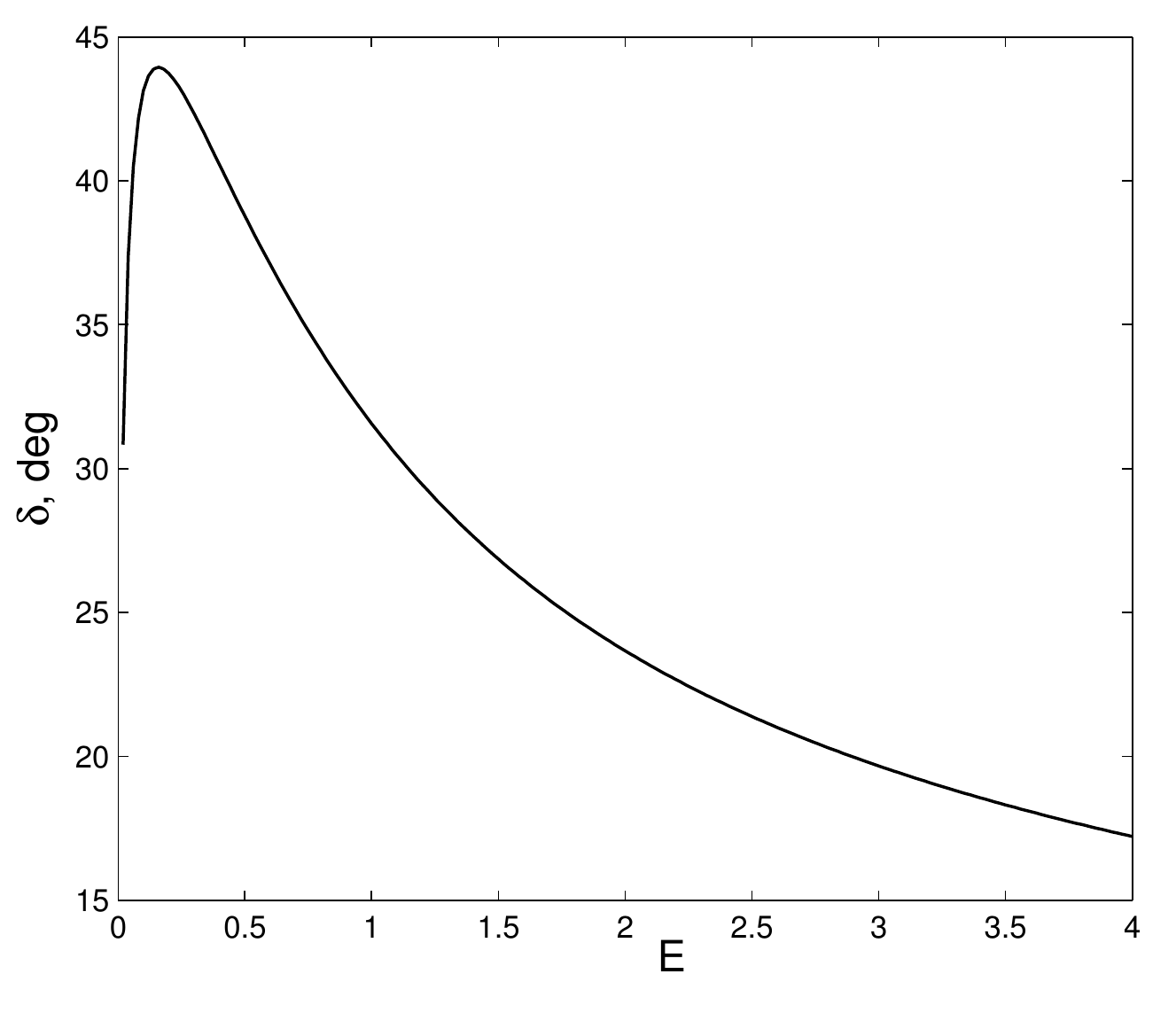}
  \includegraphics[width=0.59\linewidth]{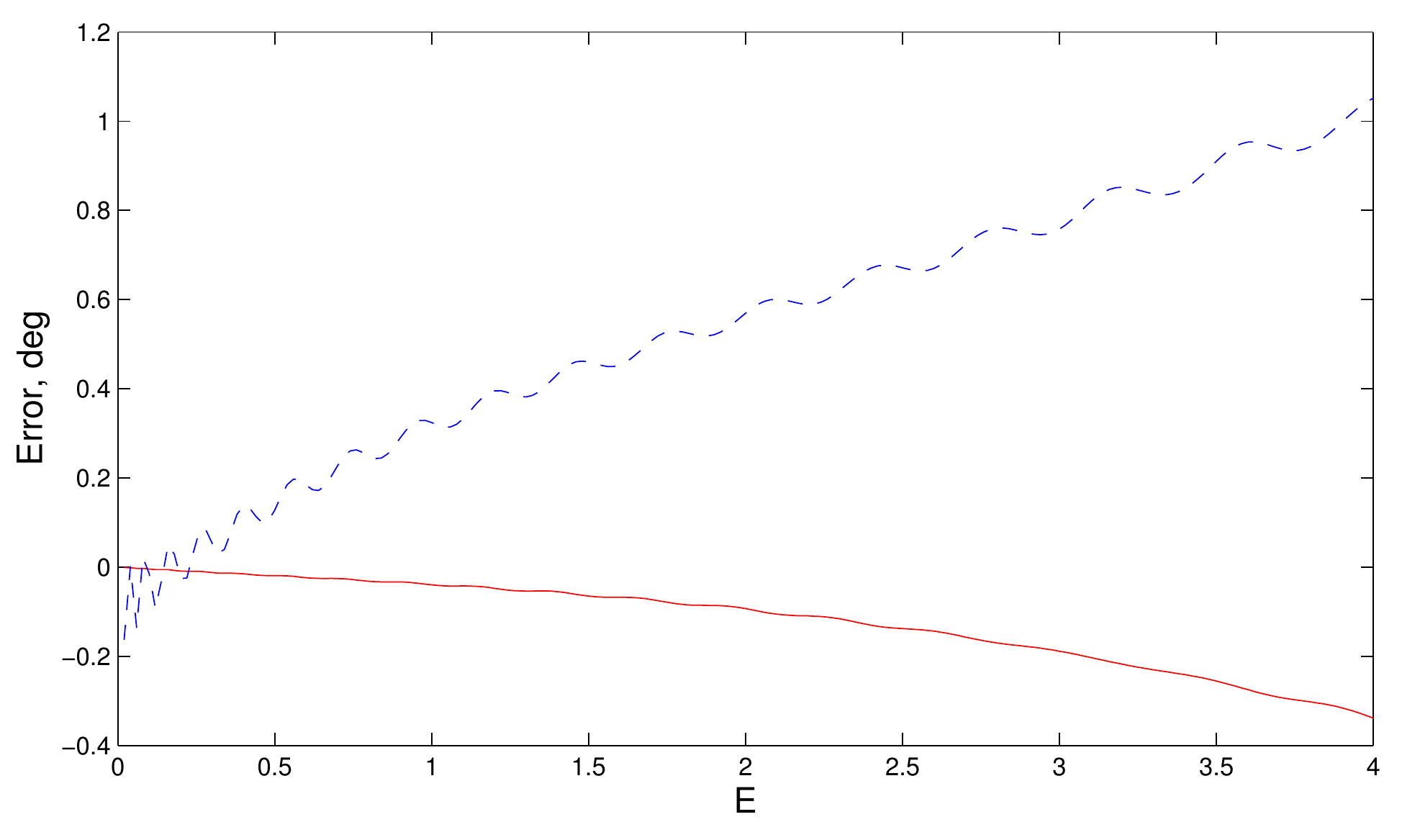}
  \caption{ Scattering phase shift of the model problem of Section
4 with a Gaussian potential (left) and the absolute error
of the scattering phase shift calculated with JM-ECS (dashed line) and
our new approach (solid line) depending on the energy of the system
(right). Both calculations were made with 100 functions in the
oscillator basis with the oscillator length $b=0.7$.}
\label{fig:pshift}
 \end{figure}

 \begin{figure}
  \includegraphics[width=0.5\linewidth]{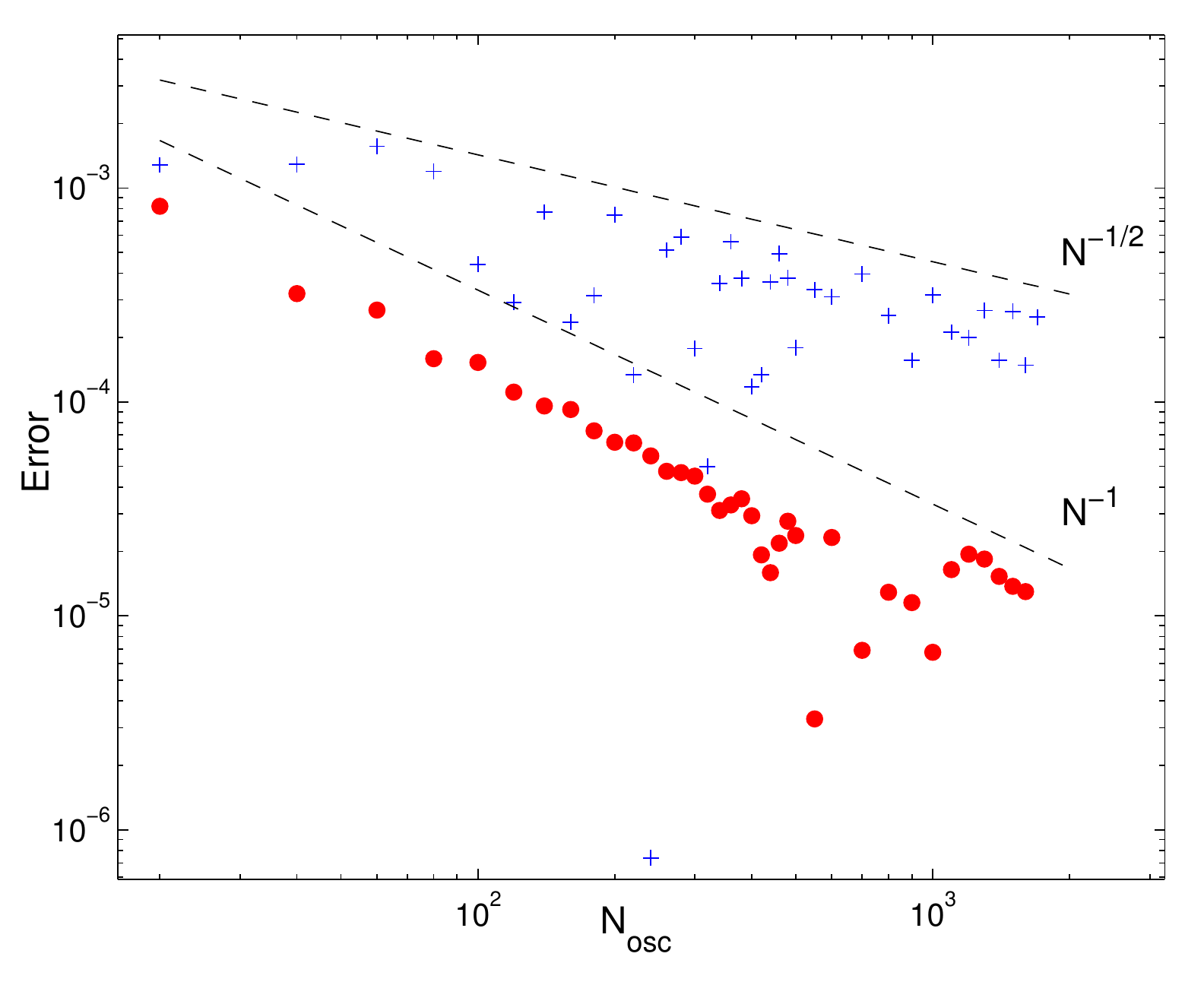}
  \includegraphics[width=0.5\linewidth]{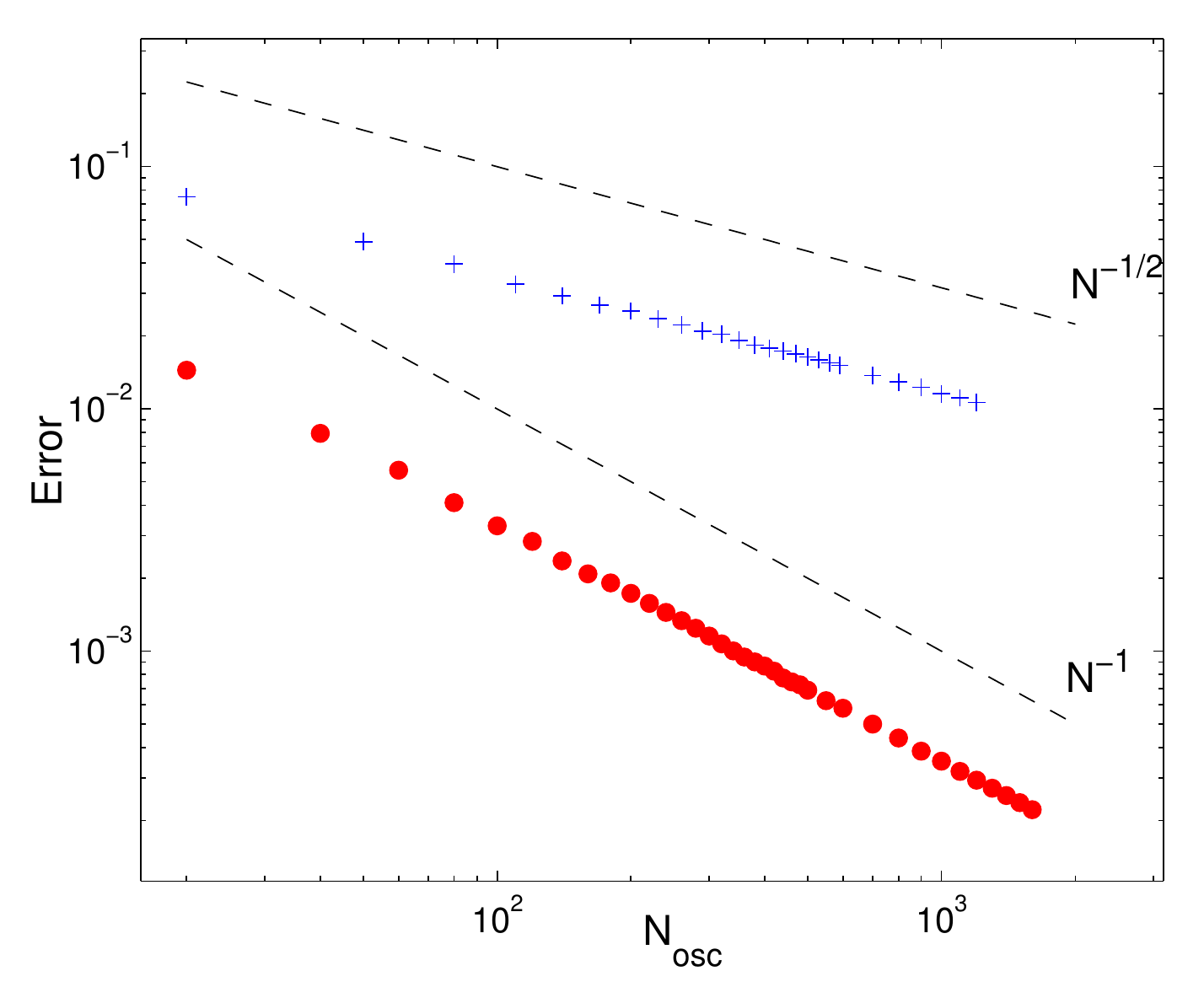}
\caption{Accuracy of the scattering phase shift calculated with JM-ECS
(crosses) and the new method (circles) depending on the size of the
oscillator basis. Calculations were made for two values of energy:
$E=0.2$, where $\delta=43.73 ^{\circ}$ (left) and $E=3$, where
$\delta=19.67 ^{\circ}$ (right)}
\label{fig:pshift_conv}
 \end{figure}

 The Gaussian shape of the interaction potential was chosen as the
 main model problem as it is particularly well adapted to the harmonic
 oscillator basis and $J$-matrix method for this potential converges for
 any oscillator length $b$. However we can also test our approach
 against other short-range potentials. But if the asymptotic behavior
 of the potential is not Gaussian then the results will converge only
 for specific values of the oscillator length, that match the range of
 the potential. Also, if the potential has a special point in the
 origin (like Yukawa potential) then the convergence will be much
 slower due to the importance of the Fourier contribution in the
 potential matrix elements. For additional convergence tests we have
 chosen two potentials with non Gaussian asymptotics and different
 behavior in the origin: Morse potential $V(x) = \exp(-2 |x/b|)-\exp(-
 |x/b|)$ and Yukawa potential $V(x) = -\exp(-|x/b|)/x$. In both
 potentials the potential range is chosen to match the oscillator
 length of the basis. Figure \ref{fig:pshift_conv_yukava} shows the
 convergence of the phase shift for these two potentials. We see that
 in high-energy region the convergence pattern is mostly similar for
 all potentials as the kinetic energy is much higher than the
 potential energy in this case. The behavior of the error is much less
 clear for low energies. The convergence rate appears to be the same
 here for both methods with Morse potential, for Yukawa potential the
 convergence rate is blurred due to high oscillations which clearly
 indicates the importance of Fourier terms. Though we can generally
 conclude that the new approach always gives better and generally less
 oscillatory result.

 \begin{figure}
  \includegraphics[width=0.5\linewidth]{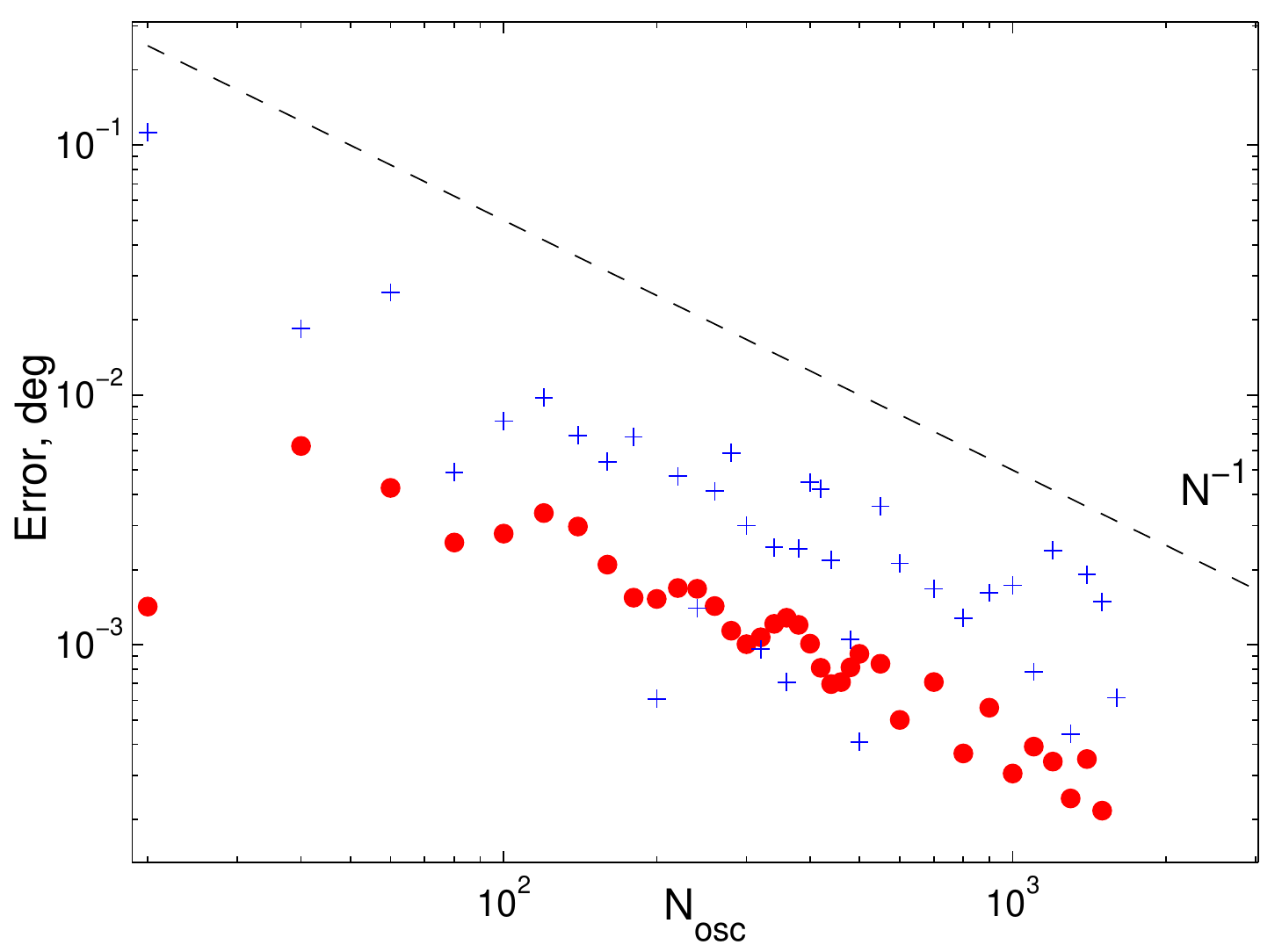}
  \includegraphics[width=0.5\linewidth]{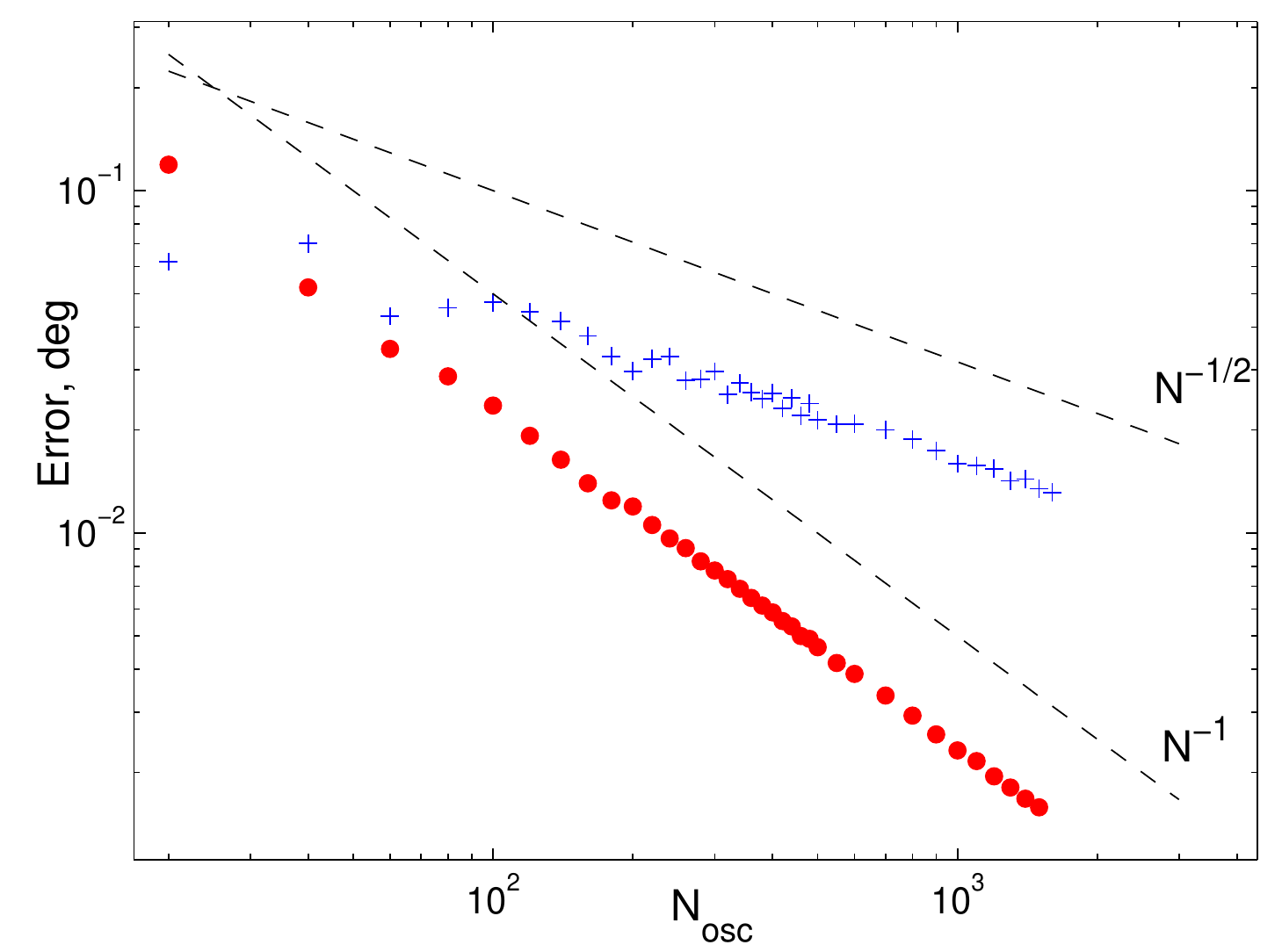}
  \includegraphics[width=0.5\linewidth]{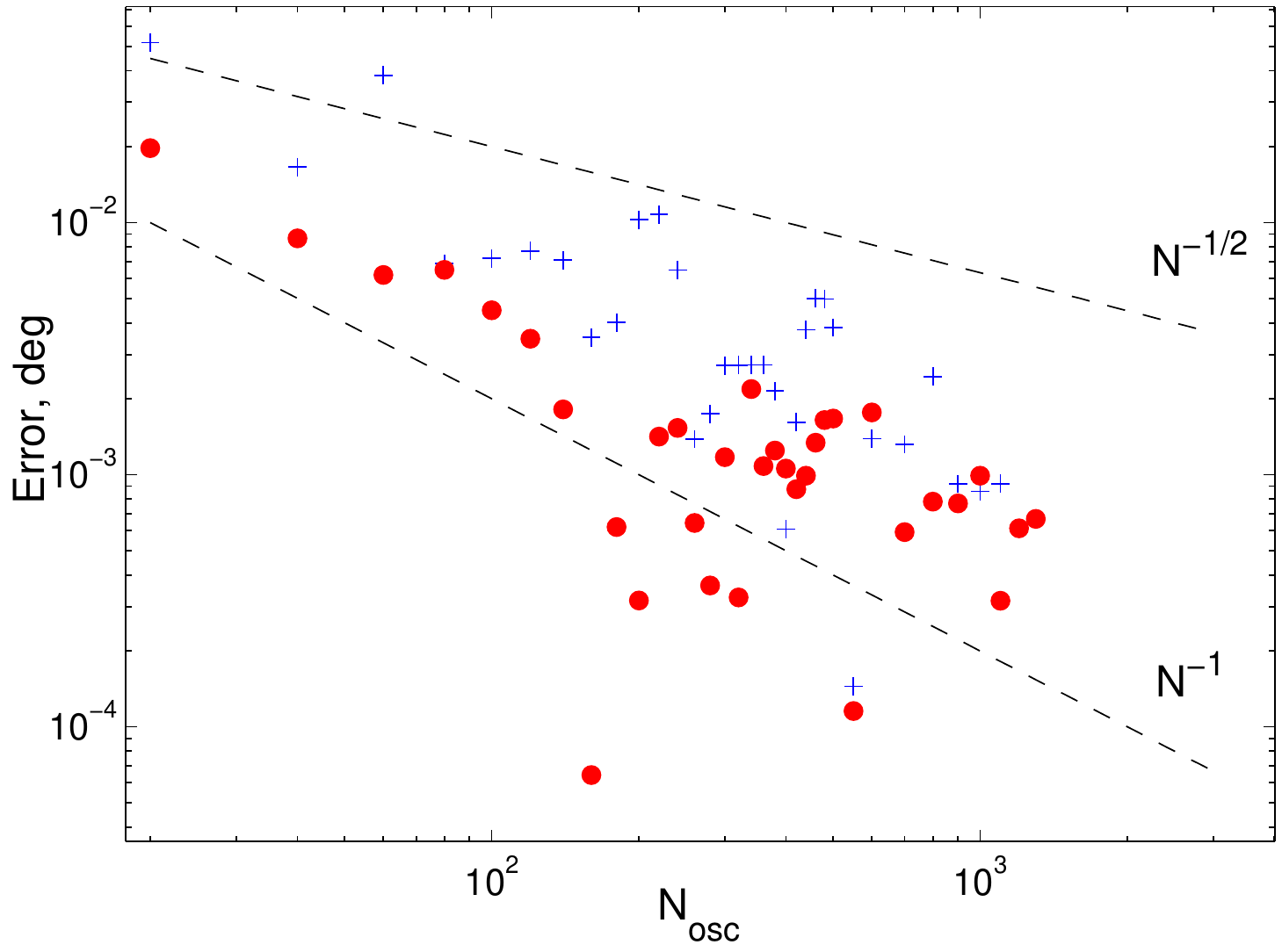}
  \includegraphics[width=0.5\linewidth]{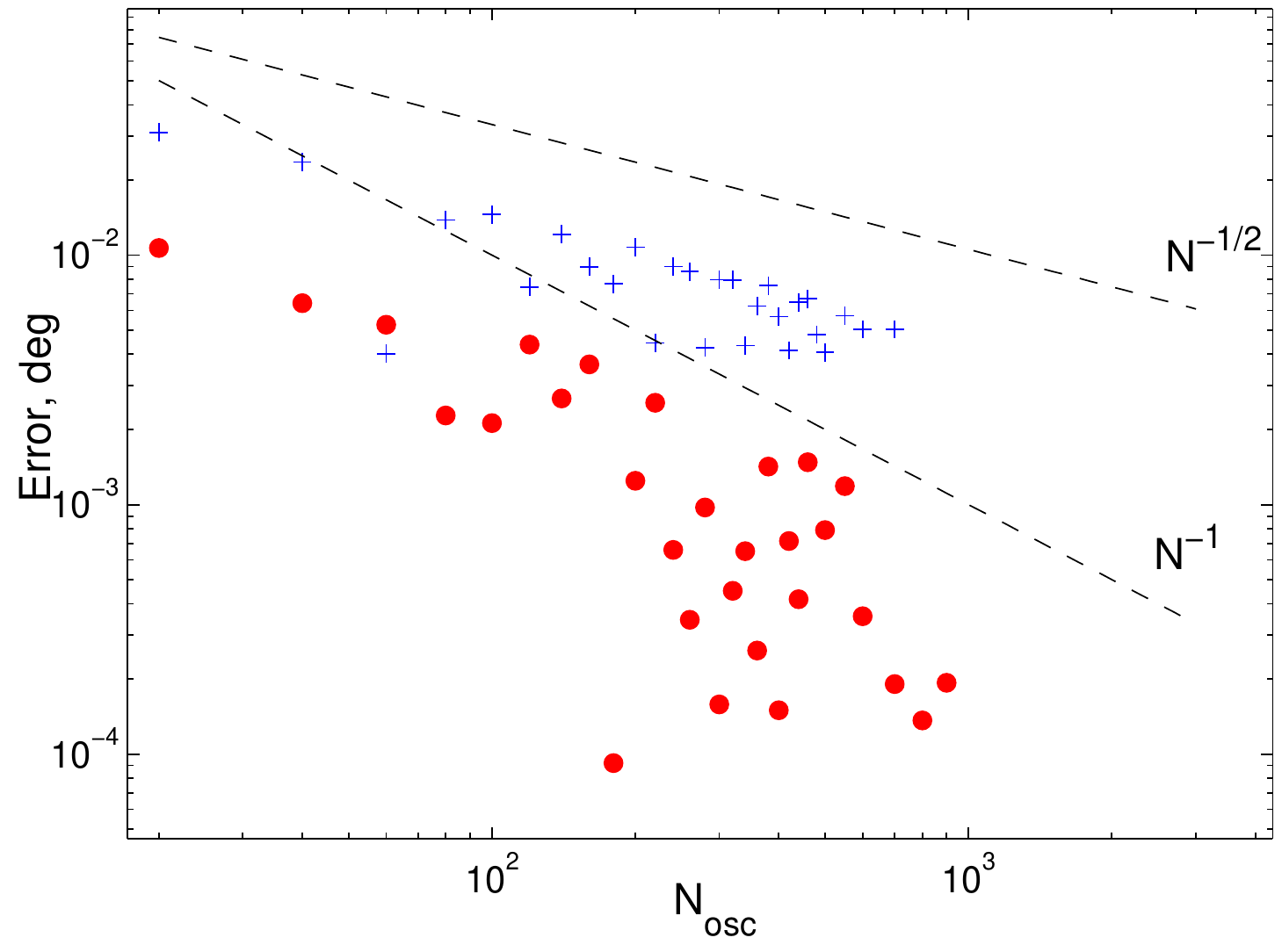}
\caption{Same as Figure \ref{fig:pshift_conv} but for the Morse potential (two upper panels) and Yukawa potential (two lower panels)}
\label{fig:pshift_conv_yukava}
 \end{figure}

\section{Extension to systems with more particles} 
The results of the previous sections can be generalized to systems of
three and more particles.  Let $\mathbf{r}_1$ and $\mathbf{r}_2$ be
two relative coordinates describing a three--body problem.  The
6D wave function can be expanded in partial waves
\begin{equation} \label{eq:2dexpansion}
\Psi(\mathbf{r}_1,\mathbf{r}_2)= \Psi(\rho_1,\rho_2,\theta_1,\theta_2,\phi_1,\phi_2) = \sum_{l_1,m_1,l_2,m_2} \frac{\psi_{l_1,l_2,m_1,m_2}(\rho_1,\rho_2)}{\rho_1 \rho_2}Y_{l_o1,m_1}(\theta_1, \phi_1)Y_{l_2,m_2}(\theta_2, \phi_2).
\end{equation}
Such an expansion is used for example to describe double ionization processes in atomic physics \cite{vanroose2006double}

Let $\psi(x,y)$ now be an infinitely
differentiable  two--dimensional radial scattering wave function that is expanded
in the bi-oscillator basis as
\begin{equation}
  \psi(x,y) = \sum_{n=0}^{\infty} \sum_{m=0}^{\infty}  c_{nm} \varphi_n(x) \varphi_m(y),
\label{eq:bioscillator}
\end{equation} 
where $x$ and $y$ should be interpreted as two radial coordinates.
The expansion coefficient is then calculated as a double integral that 
can be approximated by applying the asymptotic relation twice. First in the $x$-direction and then in the $y$-direction
\begin{equation}
\begin{aligned} \label{eq:2derrors} c_{nm}=& \int_0^\infty
\int_0^\infty \varphi_n(x)\varphi_m(y) \psi(x,y) \,dx dy \\
=&\int_0^\infty \varphi_n(x)\left[ \sqrt{2} \,
R_m^{-1/2}\psi(x,R_m) + (b^4/6) R_m^{-3/2} \partial_{yyy} \psi(x,R_m)
+ \mathcal{O}(R_m^{-5/2})\right]dx \\
 =& 2\, (R_m
R_n)^{-1/2}\psi(R_n,R_m) + \frac{b^4 \sqrt{2}}{6} R_n^{-3/2}
R_m^{-1/2} \partial_{xxx} \psi(R_n,R_m) + \frac{b^4 \sqrt{2}}{6}
R_n^{-1/2}R_m^{-3/2}\psi_{yyy}(R_n,R_m) \\ & +
\frac{b^8}{36}R_n^{-3/2}R_m^{-3/2}\partial_{xxx,yyy} \psi(R_n,R_m) +
\mathcal{O}(R_n^{-1/2}R_m^{-5/2}) + \mathcal{O}(R_m^{-1/2}R_n^{-5/2}).
\end{aligned}
\end{equation}
It is important to note that the accuracy depends on both indices $n$
and $m$.  For example, when the low order approximation is used for
both integrals, there will be error terms of the order
$R_m^{-1/2}R_n^{3/2}$ and $R_n^{-1/2}R_m^{3/2}$.  Along the diagonal,
i.e.  when $R_m = R_n$, it is accurate up to order
$R_n^{-1/2}R_m^{-3/2} = R_n^{-2} = \mathcal{O}(n^{-1})$, which is of
the order $h_n^2$, with $h_n$ the distance between the classical
turning points.  For the higher--order
approximation the error term along the diagonal, where $R_n=R_m$, will
be $R_n^{-3} = \mathcal{O}(n^{-3/2})$. In the left panel of Figure
\ref{fig:coeff_2d} we compare the exact expansion coefficients with
the first order and second order approximation along the diagonal for
$n=m$.  The results show an improvement with the higher--order
expression over the first order approximation.

However, when one of the indices remains small, for example $m$, then
the error becomes asymptotically $\mathcal{O}(R_n^{-1/2})$, as
expected from \eqref{eq:2derrors}.  As $n$ becomes large, all error
terms with a higher--order accuracy decay until the term $O(R_n^{-1/2}
R_m^{-5/2})$ dominates for the higher--order formula, or $O(R_n^{-1/2}
R_m^{-3/2})$ for the low order accuracy. This is shown in the right
panel of Figure \ref{fig:coeff_2d}.
\begin{figure}
  \begin{center}
	\includegraphics[width=0.45\linewidth]{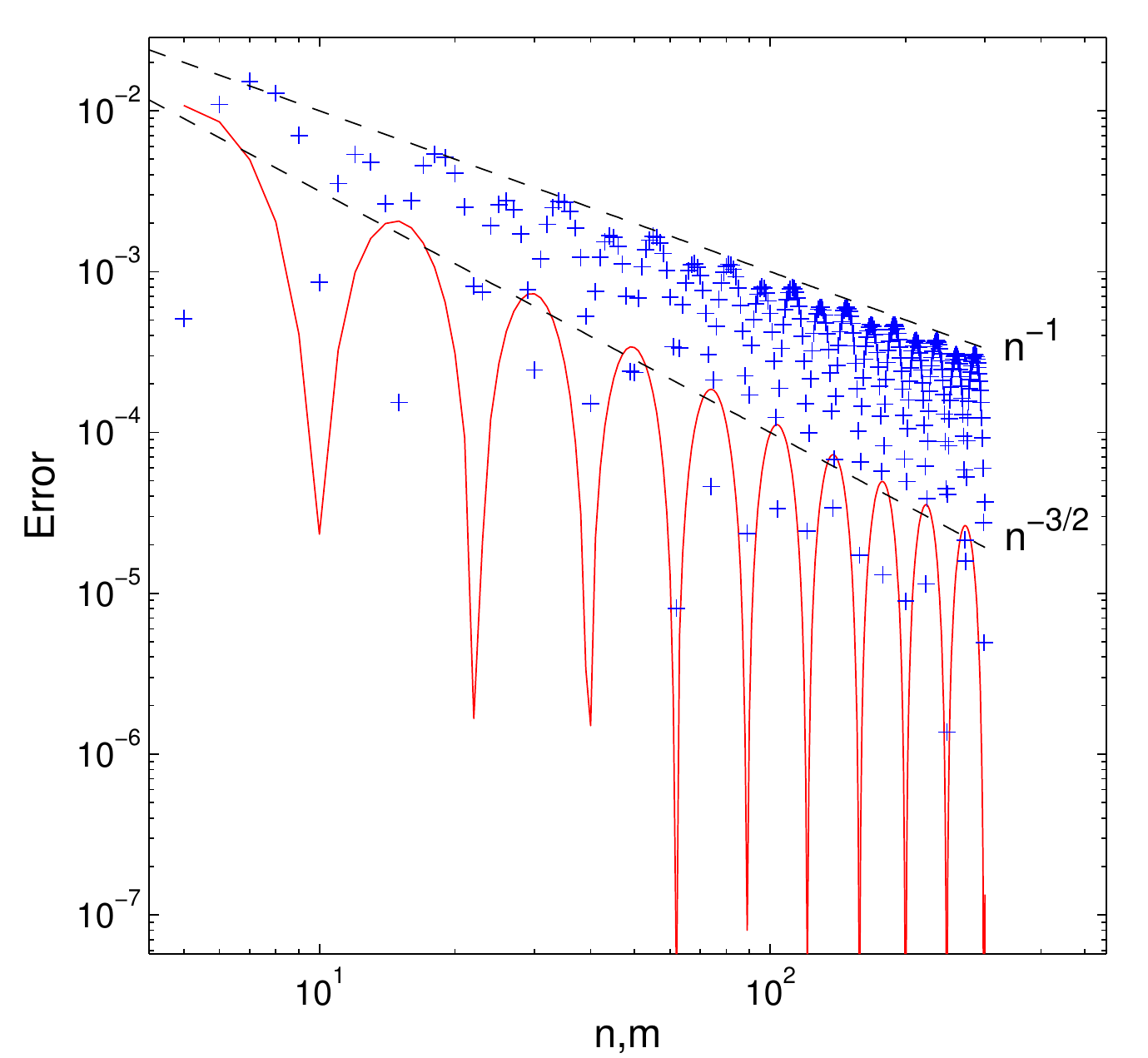}
\includegraphics[width=0.45\linewidth]{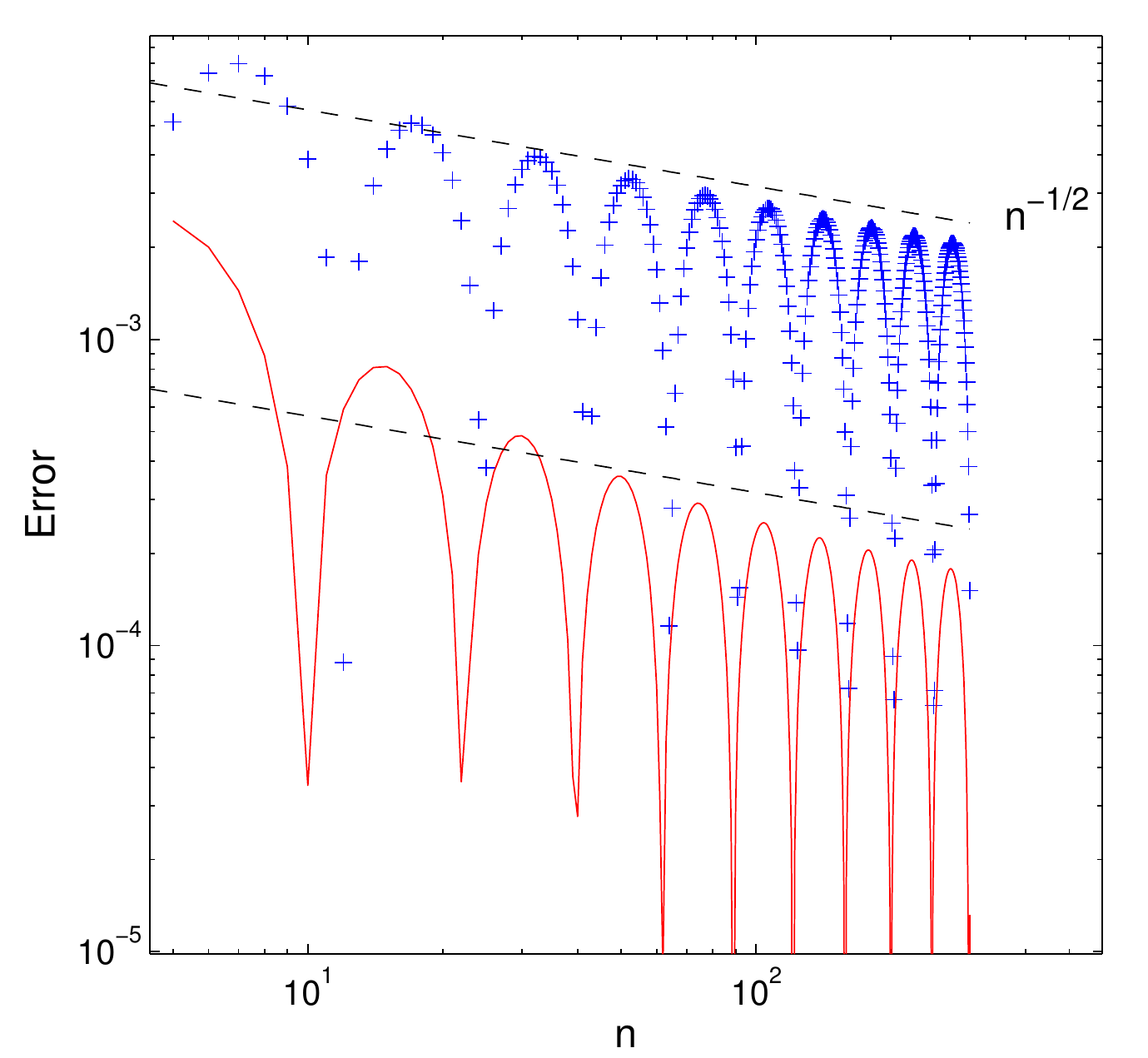}
      \end{center}
      \caption{Convergence of the asymptotic formula for
two-dimensional radial functions. The error is calculated by comparing the
exact expansion coefficient with the asymptotic approximations. The
function $f(x,y) = \sin(kx) \sin(ky)$ along the diagonal (n=m, left)
and along the line with fixed m=20 (right).  Crosses show the
first order and solid lines correspond to the second order asymptotic
expansion coefficient ($k=1$, $b=1$). }
\label{fig:coeff_2d}
\end{figure}

\subsection{Three--body scattering problems} 
In a similar way as in the previous sections we can use the asymptotic
formulate to build a hybrid representation that can solve the radial
equation of two radial variables that arises after an expansion of a full
three--body scattering problem equation in spherical harmonics \eqref{eq:2dexpansion}. This typically leads to a set of coupled two--dimensional radial 
partial waves. A diagonal block of this coupled set reads
\begin{equation}
  \left(-\frac{1}{2}\frac{d^2}{dx^2} -
    \frac{1}{2}\frac{d^2}{dy^2} + \frac{l_1(l_1+1)}{2x^2} +
    \frac{l_2(l_2+1)}{2y^2} + V(x,y) - E\right) \psi(x,y) = \chi(x,y)
\label{eq:2d_scattering}
\end{equation}
with boundary conditions  $\psi(x,0)=\psi(0,y)=0$.

The solution of this equation can be represented in a bi-oscillator
basis, (\ref{eq:bioscillator}), using two indices $n$ and $m$. Then
similar to the two--body case we use approximate oscillator
coefficients in the asymptotic region. But there are now three
asymptotic regions: first, the region where $n$ is large, second, the
region where $m$ is large and finally, the region where both $n$ and
$m$ are large.  For each of these regions we define approximate
oscillator coefficients that take the form:
\begin{align}
 c_{nm} \approx \grave c_{nm} := & \sum_{i} U_{ni} \int_0^{\infty} \varphi_m(y)
  \psi(R_i,y) \,dy + \mathcal{O}(R_n^{-5/2})  \quad \text{where} \quad n \gg 1\\
 c_{nm} \approx \acute c_{nm} := & \sum_{i} U_{mi} \int_0^{\infty} \varphi_n(x)
  \psi(x,R_i) \,dx + \mathcal{O}(R_m^{-5/2}) \quad \text{where} \quad m \gg 1\\
c_{nm} \approx  \check c_{nm} := & \sum_{i,j} U_{ni} U_{mj} \psi(R_i,R_j) +
  \mathcal{O}(R_n^{-1/2}R_m^{-5/2}) +
  \mathcal{O}(R_m^{-1/2}R_n^{-5/2}) \quad \text{where} \quad n,m \gg 1,
\end{align} where the error terms in the last expression are explained
in (\ref{eq:2derrors}).

Then the coefficient matrix $c_{nm}$ of hybrid representation is divided into four blocks
corresponding to different regions
\begin{equation} 
  c_{nm} = \left(\begin{array}{lll|lll} c_{0\,0} &
      \ldots &c_{0\,M-1} & \acute c_{0\,M} & \ldots &\acute c_{0\,{K'}}\\
      c_{1\,0} & \ldots &c_{1\,M-1} & \acute c_{1\,M} & \ldots &\acute
      c_{1\,{K'}}\\ 
\vdots & & \vdots& \vdots & &\vdots\\
c_{N-1\,0} & \ldots &c_{N-1\,M-1} & \acute c_{N-1\,M} &
      \ldots & \acute c_{N-1\,K'} \\ \hline \grave c_{N\,0} & \ldots &\grave
c_{N\,M-1} & \check
      c_{N\,M} & \ldots &\check c_{N\,K'}\\ \grave c_{N+1\,0} & \ldots &\grave
      c_{N+1\,M-1} & \check c_{N+1\,M} & \ldots &\check c_{N+1\,K'}\\ 
\vdots & &\vdots & \vdots & &\vdots\\
 \grave c_{K\,0} & \ldots &\grave c_{K\,M-1} &
\check
      c_{K\,M} & \ldots &\check c_{K\,K'}\\
\end{array}\right),
\end{equation} 
where we define $N$ and $M$ as sizes of the oscillator
bases in each direction and $K$ and $K'$ are the total number of
variables in each direction.

To build a linear system corresponding to (\ref{eq:2d_scattering}), we
reshape the matrix $c_{nm}$ to a vector. Then the Hamiltonian matrix of a 2D
problem is constructed as a Kronecker sum of the two-body Hamiltonian's
constructed as in (\ref{eq:combined}) and a two-body potential matrix. In the
case of our approximate oscillator representation this takes the form
\begin{equation}
 \check H^{(osc)}_{2D} = \check H^{(osc)}_{1} \otimes (U^y W^y) + (U^x W^x) \otimes
\check H^{(osc)}_{2} + (U^x \otimes U^y) V_{12}^{(fd)} (W^x \otimes W^y)
\end{equation}
so the Kronecker sum contains the approximated unity operator instead
of the real one. Where $U^x$ and $W^x$ are the transformation matrices
for the $x$ coordinate and similarly for the $y$ coordinate.  The
total size of the Hamiltonian matrix is $K \times K^\prime$.

\subsection{Numerical results}
We first evaluate the method with the help of a model Helmholtz
problem with a constant wave number or energy $E=k^2/2$ and for
$l_1=0$ and $l_2=0$.  The equation is then
\begin{equation} \left(-\frac{1}{2} \Delta - E\right) \psi(x,y) =
\varphi_0(x)\varphi_0(y), 
\label{eq:simple_2d_problem}
\end{equation} with boundary conditions $\psi(x,0)=0$ and
$\psi(0,y)=0$.  
The form of the right--hand side here was chosen for simplicity as it results
in a vector $(1,0,0,\dots 0)$ in bi-oscillator representation.
This problem is exactly solvable with the help of the
Greens function
\begin{align} G(x,y;x',y') = & \frac{i}{2} \left(
H_{0}^{(1)}(k\sqrt{(x-x')^2+(y-y')^2}) -
H_{0}^{(1)}(k\sqrt{(x+x')^2+(y-y')^2}) \right. \\
-&\left. H_{0}^{(1)}(k\sqrt{(x-x')^2+(y+y')^2}) +
H_{0}^{(1)}(k\sqrt{(x+x')^2+(y+y')^2}) \right), 
\end{align} where $H_0^{(1)}$ is a 0-th order Hankel function of the
first kind.  The scattering solution is then
\begin{equation} \psi(x,y) = \int_0^\infty \!\!\int_0^\infty \varphi_0(x')\varphi_0(y')
G(x,y;x',y')\; x'y'\; dx'\; dy'
\label{eq:green_integral}
\end{equation} 

For simplicity we are not going to compare any scattering information extracted
from the wave function as it involves additional operations like surface
integration, which can lead to additional loss of accuracy and need to be
studied separately. We will compare the values of the wave function in a fixed
spatial point. As the spatial grid is different for every size of the
oscillator basis we can not ensure that any fixed spatial point lies exactly on
the grid point at every calculation, so we need to interpolate. We have used
a cubic spline interpolation, and from comparison to the other interpolation
methods we can expect that the additional error introduced by this operation is
negligible compared to the errors of the method (of course, that is only if
the function is smooth enough, that means not very high values of $E$).

As we do not have any potential in this model problem, the convergence of the
proposed hybrid oscillator representation will depend mainly on the accuracy of
the asymptotic relations used. This means that we can choose the size of the
spatial domain on which we solve the two--dimensional radial problem as small as possible but not
smaller than the region spanned by the exact oscillator basis. As the maximal
size of the oscillator basis we use in this calculations is $350$ and the
oscillator length was chosen as $b=0.7$ ($R_N \approx 26$) the size of the
spatial domain was chosen as $30$ dimensionless length units of the real
grid and additional $10$ units of the ECS layer. The value of the wave
function was extracted at the point $(28,28)$ when we increase the basis size
simultaneously and at $(28,12)$ when we keep the basis size fixed in the
$y$-dimension.

In Figure \ref{fig:conv_wf_2d} we compare the numerical solution of
the linear system with the one obtained from~\eqref{eq:green_integral}
which we consider as exact. We see that, similar to
Figure~\ref{fig:coeff_2d}, if we expand the basis only in one
direction the convergence is of the order $N^{-1/2}$.
\begin{figure}
  \begin{center}
	\includegraphics[width=0.45\linewidth]{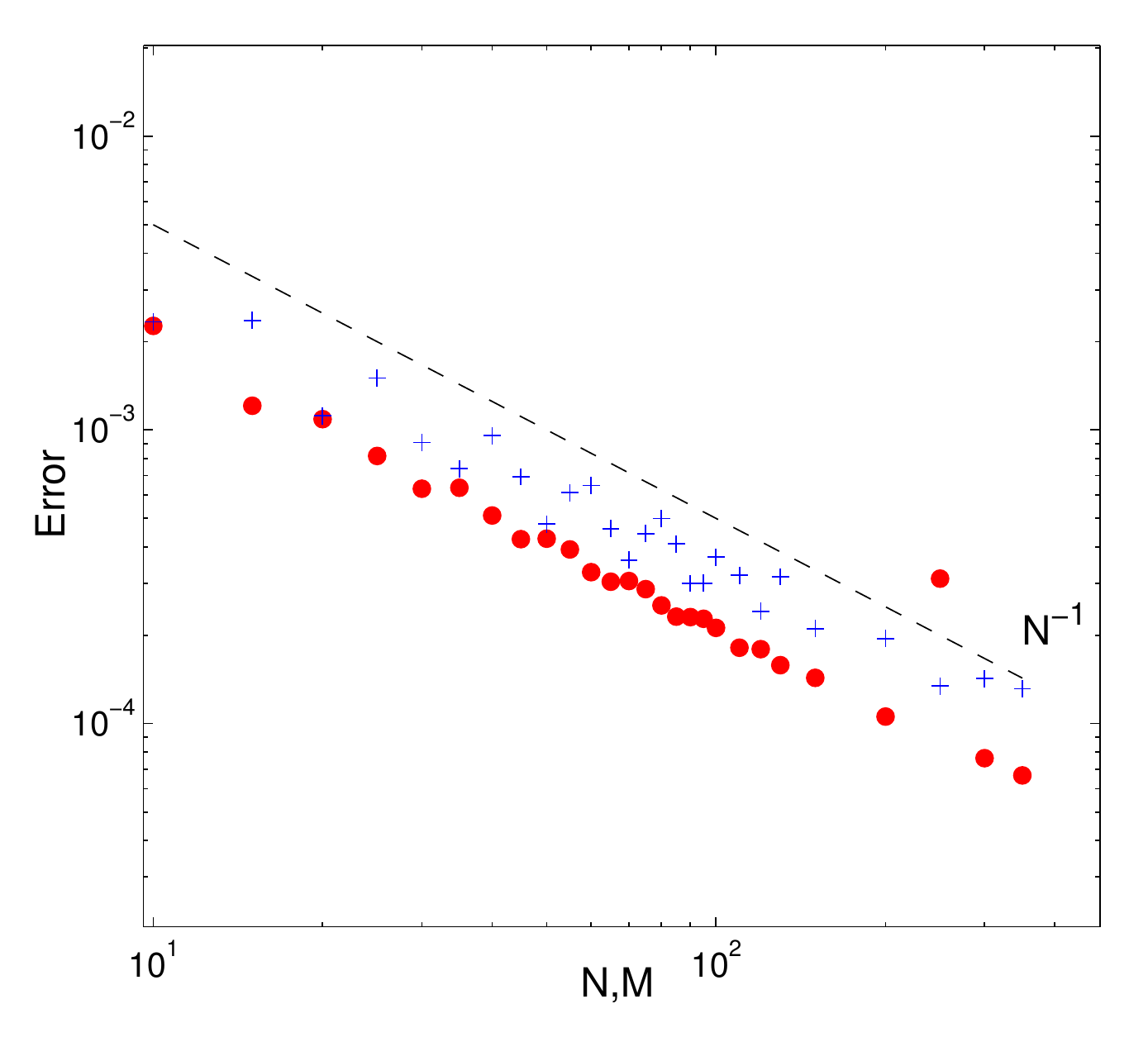}
\includegraphics[width=0.45\linewidth]{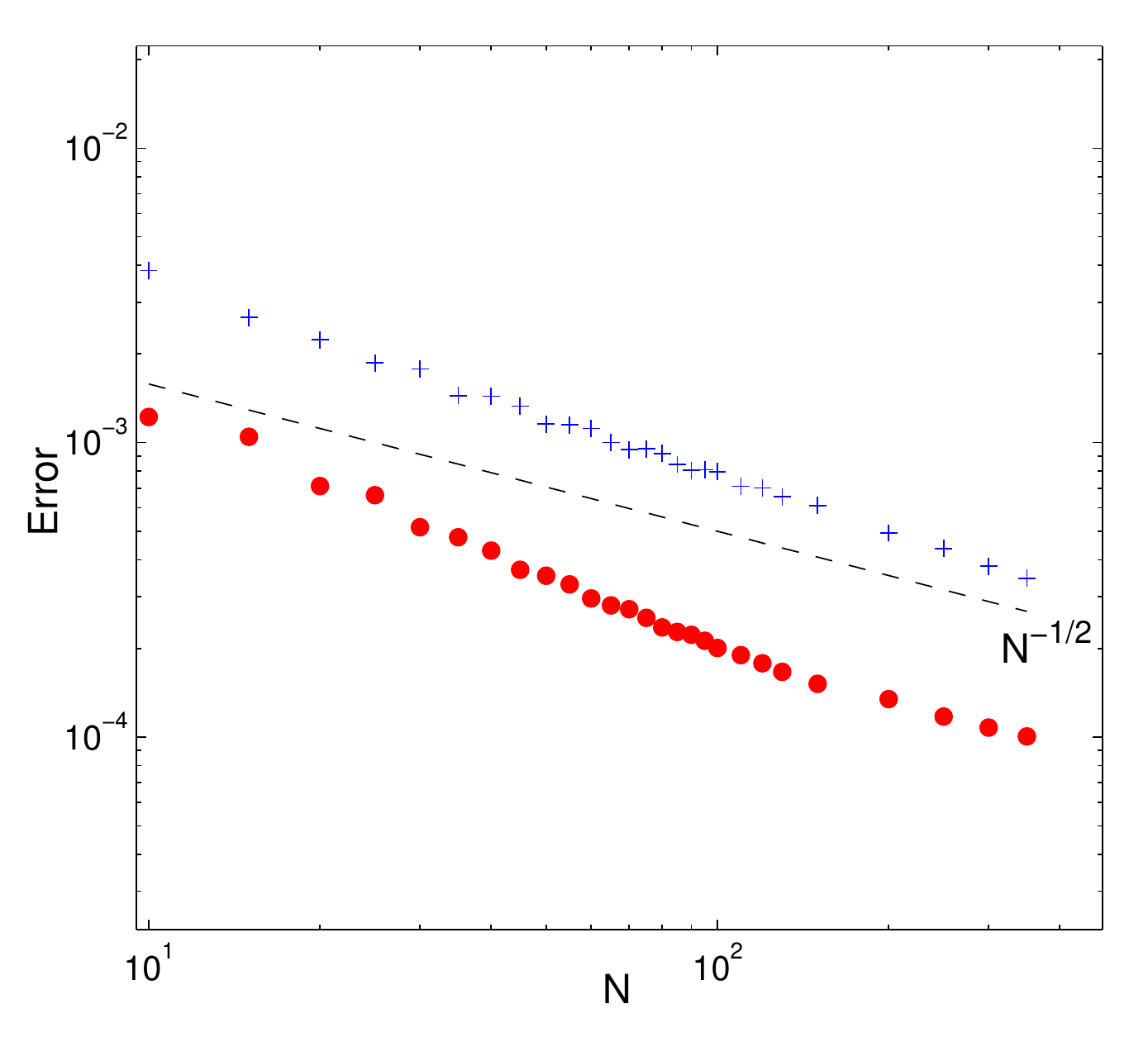}
      \end{center}
      \caption{Convergence of the two-dimensional scattering wave
function for the problem in Eq.~\eqref{eq:simple_2d_problem} at the fixed
spatial point. Left figure represents the case
where the basis was increased in two dimensions simultaneously ($N =
M$), right figure shows the convergence with a fixed $M = 50$. Both
calculations are made with $b=0.7$ and $E=2$.}
\label{fig:conv_wf_2d}
\end{figure}

Next we test the method on a model potential scattering problem
described by (\ref{eq:2d_scattering}) with zero partial angular momenta
and the Gaussian interaction potential in the form
\begin{equation}
 V(x,y) = -3\mathrm{e}^{-x^2} - 3\mathrm{e}^{-y^2} + 10
\mathrm{e}^{-x^2-y^2}
\label{eq:2d_pot}
\end{equation}
The one-body potential, here $V_1(x) = -3\exp(-x^2)$, can support one bound
state. This means that using the potential \eqref{eq:2d_pot} we can model a
three-body breakup problem using the right--hand side in the form
\begin{equation}
 \chi(x,y) = -(V(x,y)-V_1(x))f_0(x)\hat j_0(y),
\end{equation}
where $f_0(x)$ is the wave function of the bound state, $\hat j_0(y)$ is the
Riccati-Bessel function and together they represent the initial state of the
system.
For the considered problem there are no analytical results to compare with, so
on
the Figure
\ref{fig:surf_2d} we present the solution of the linear system
for $c_{nm}$ and the spatial wave function, reconstructed from
$c_{nm}$ by applying the transformation matrix $W \otimes W$. On both figures
we can clearly recognize the patterns of elastic scattering (the plane waves
along the axes) and the breakup (the radial waves with lower frequency as part
of the energy was spent to break the bound state). The in-depth analysis of the
three--body scattering results is the subject of future work.
\begin{figure}
  \begin{center}
	\includegraphics[width=0.49\linewidth]{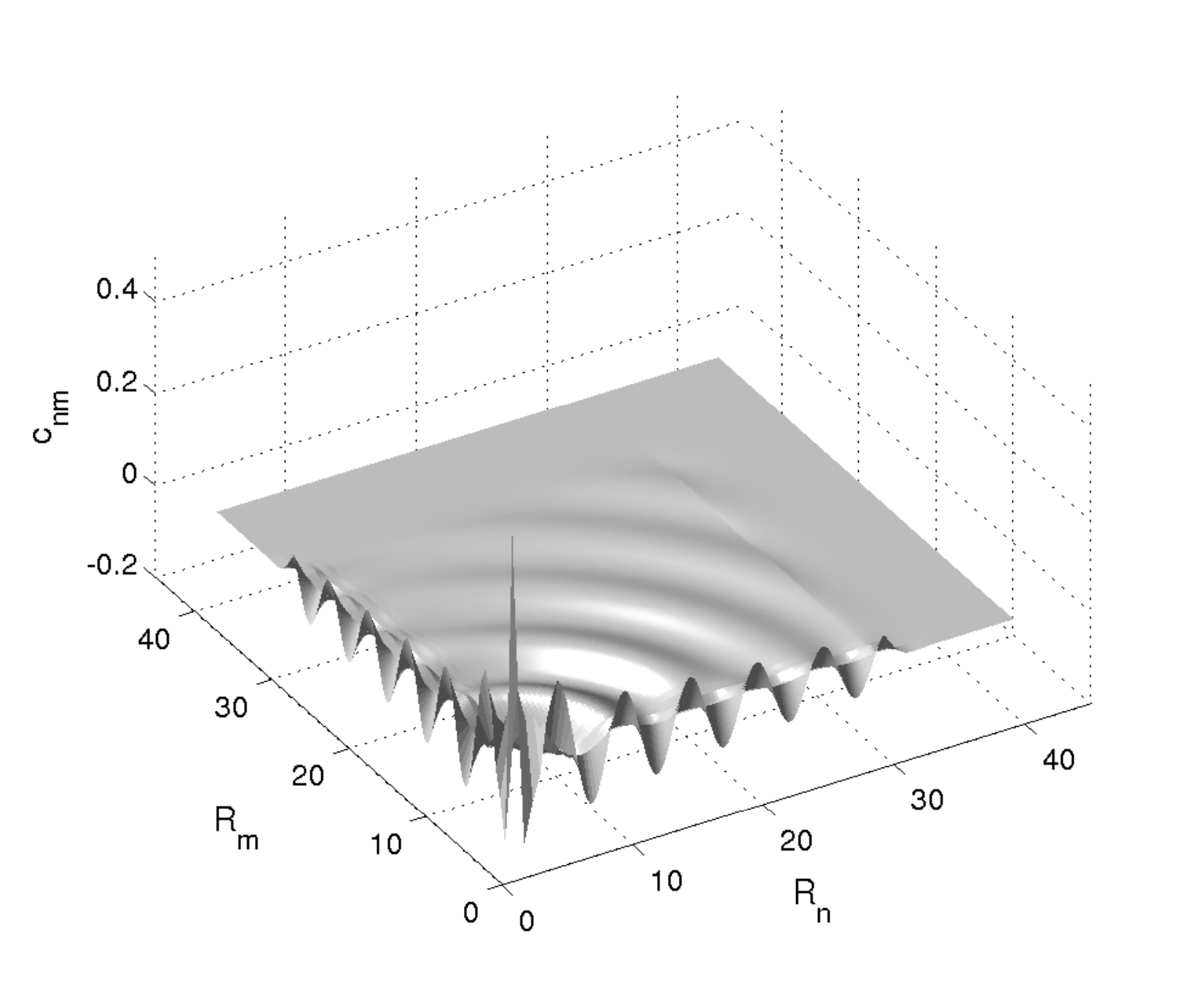}
\includegraphics[width=0.49\linewidth]{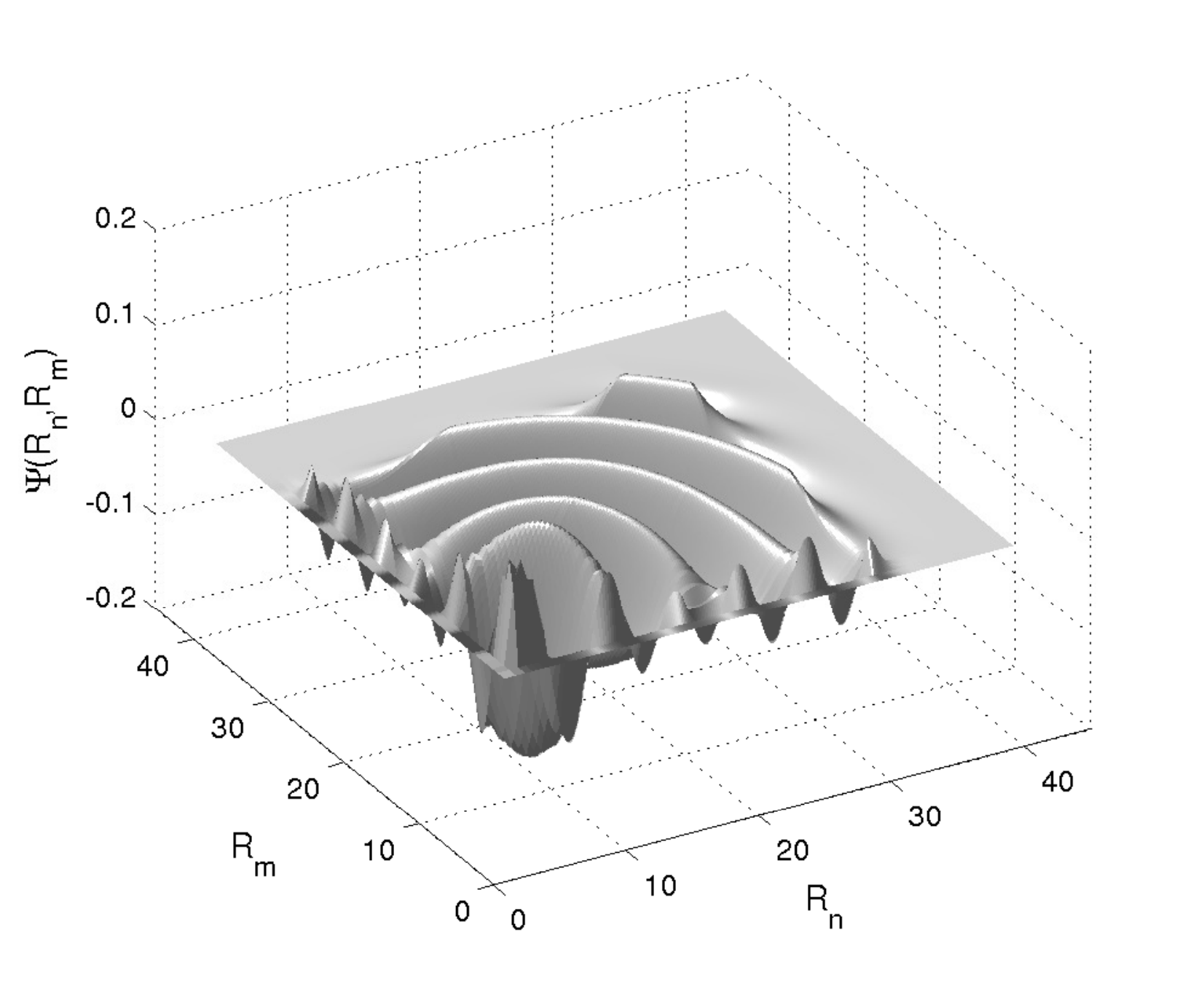}
      \end{center}
      \caption{The solution of three--body scattering problem in oscillator representation
(left figure) and the reconstructed spatial wave function (right
figure). The calculation is made with $b=1$ and $E=1$. The problem has Dirichlet boundaries on all sides. }
\label{fig:surf_2d}
\end{figure}

\section{Conclusions and outlook} This paper focuses on scattering
calculations in the oscillator representation, where the solution is
expanded in the eigenstates of the harmonic oscillator.  The
oscillator representation is not the most natural representation to
describe scattering processes since it involves a basis set that is
designed to describe bound states. This leads to a rather low
convergence and may result in a linear system with very large dense
matrix. It is often more natural to describe a scattering process with
the help of a grid representation.  These grid--based calculations
have been very successful in describing scattering and breakup
processes in atomic and molecular physics. The Helmholtz equation is
also efficiently solved on these grids.

However, internal structure of the nuclear clusters and other
many--particle systems are efficiently described by such an oscillator
basis, since the bottom of the potential can be mimicked by the
oscillator potential leading to an efficient description of the
internal structure.

In this paper we combine the advantages of grid--based calculations
with those of the oscillator representation.  The method was
originally proposed in \cite{PhysRevC.82.064603}, as a further
development of the $J$-matrix or algebraic method for
scattering. There the method combined the grid and oscillator
representation with a low--order asymptotic formula. In this paper we
have improved this matching with a higher--order approximation and
this brings the overall error down to the level of the discretization
error of the finite--difference grid.

Although a similar asymptotic formula appeared earlier in the work of
S. Igashov \cite{igashov_jmatrix}, we believe that the asymptotic
formula presented in the current paper is more generally
applicable. Furthermore, we have used the asymptotic formulas to
improve the accuracy and convergence of the hybrid simulation method.
The convergence of the method is illustrated with various examples
from two--body and three--body scattering.

In preparation of this papers, our initial efforts involved application of the
strategy showed in Figure \ref{fig:coupling} with the higher--order
formula. However, this required the use of a forward and backward
stencils to estimate the third derivative in the first grid point of
the finite difference representation.  This strategy, however, gives
rise to eigen modes with a negative energy localized near the
interface. This destroys the positive definiteness of the Laplace
operator. These modes are avoided in the proposed method that uses a
symmetric stencil around the matching point. This gives satisfactory
convergence.

We have applied the asymptotic technique to radial oscillator state
that are based on Laguerre polynomials.  Similar results can be
derived for the 1D oscillator states that are based on the Hermite
polynomials.  These Hermite polynomials are closely related to
Gaussian Type Orbitals (GTO) that are frequently used in computational
quantum chemistry \cite{mcmurchie1978one} .  It is worth to explore if
the asymptotic formulas make it possible to combine GTO's with grid
based calculations to describe molecular scattering processes.

In the future further research is necessary on asymptotic expressions
of products of two functions.  This will involve convolution
integrals. Better expressions for products will also help to take into
account asymptotic behavior of the potentials in the grid representation.

\section{Acknowledgments}
We acknowledge fruitful discussions with F. Arickx and J. Broeckhove
and we are grateful to B.  Reps and P. K{\l}osiewicz for reading the
manuscript.  We also thank the anonymous referees for their
suggestions.  We are thankful for support from FWO--Flanders with
project number G.0120.08.  

\bibliographystyle{elsarticle-num-names}
\bibliography{Refs}

\end{document}